\documentclass[11pt]{article}
\usepackage[utf8]{inputenc}

\usepackage{amsmath}
\usepackage{amsfonts}
\usepackage{amssymb}
\usepackage{graphicx}
\usepackage{subcaption}
\usepackage{color}
\usepackage{float}
\usepackage{geometry}
\usepackage{natbib}%
\usepackage{booktabs}
\usepackage{url}
\usepackage[colorlinks=true,linkcolor=blue,citecolor=blue]{hyperref}%
\numberwithin{figure}{section}
\usepackage{eurosym}

\providecommand{\U}[1]{\protect\rule{.1in}{.1in}}
\newtheorem{theorem}{Theorem}[section]

\newtheorem{definition}{Definition}[section]

\newenvironment{proof}[1][Proof]{\noindent\textbf{#1.} }{\ \rule{0.5em}{0.5em}}
\numberwithin{equation}{section}

\numberwithin{equation}{section}

\newcommand{\be}{\begin{equation}}
\newcommand{\ee}{\end{equation}}

\newcommand{\bq}{\begin{eqnarray}}
\newcommand{\eq}{\end{eqnarray}}

\setcitestyle{authoryear,round}
\usepackage{mathrsfs}
\usepackage{setspace}
\usepackage{esvect}

\providecommand{\JEL}[1]{\textbf{JEL:} #1}

\begin{document}



\title{Contract Structure and Risk Aversion in Longevity Risk Transfers\thanks{We extend our gratitude for the valuable feedback received during various academic events, including the 2023 and 2024 International Conference on Actuarial Science, Quantitative Finance and Risk Management, the Actuarial Research Conference 2023, the Eighteenth International Longevity Risk and Capital Markets Solutions Conference, and the Statistical Society of Canada 2024 Annual Meeting, as well as from the engaging seminars held at the University of New South Wales, York University, the University of Nebraska Lincoln, Chongqing University, the Southern University of Science and Technology, and Monash University. The usual disclaimer applies.}}

\author{David Landriault\thanks{Department of Statistics and Actuarial Science, University of
Waterloo, Waterloo, ON, N2L 3G1, Canada (
david.landriault@uwaterloo.ca)}
\and Bin Li\thanks{Department of Statistics and Actuarial Science, University of
Waterloo, Waterloo, ON, N2L 3G1, Canada (bin.li@uwaterloo.ca)}
\and Hong Li
\thanks{Department of Econometrics and Business Statistics, Monash University, Caulfield, Victoria, Australia \\(hong.li1@monash.edu)}
\and Yuanyuan Zhang
\thanks{Department of Statistics and Actuarial Science, University of
Waterloo, Waterloo, ON, N2L 3G1, Canada (y2765zha@uwaterloo.ca)}}


\date{}

 \maketitle
\begin{abstract}

This paper develops an economic framework for optimal longevity risk transfer between a buyer and a seller with different risk aversions. We compare static (long-dated, pre-committed) and dynamic (short-dated, rolled) longevity swaps in a Stackelberg game. We find that static contracts are preferred when the buyer is more risk averse, while dynamic contracts are preferred when the seller is more risk averse. For the capital-market setting, we extend the benchmark by introducing seller-side ambiguity about the mortality distribution and robust max-min valuation.
Even moderate ambiguity can eliminate the market for static swaps, while dynamic designs remain viable.  {We then extend the analysis to index-based swaps with basis risk: relative to indemnity swaps, optimal loadings are lower and gains are smaller for both parties, though the static–dynamic preference pattern is unchanged.}
\JEL{G22, G14, G41}

\end{abstract}

\section{Introduction}

The global aging population has intensified the need for defined‑benefit pension plans, life insurers, and reinsurers to effectively manage and transfer their exposure to longevity risk -- the risk that individuals live longer than expected.  This has led to the emergence of longevity risk transfer as a crucial financial tool.   {This paper studies how differences in risk aversion of the two counterparties of longevity risk transfer shapes the equilibrium choice between static (long‑dated, pre‑committed) and dynamic (short‑dated, rolled) hedging in a Stackelberg framework.}

To date, most transactions occur in the reinsurance channel, where capacity is provided by globally diversified reinsurers and where contracts are typically \textit{long-term}. Although reinsurance has been the dominant outlet over the past two decades, it alone cannot absorb global longevity exposure, which far exceeds its capacity \citep{blake2017longevity}. Complementary participation by capital‑market investors is therefore crucial, both to add capacity for hedgers and to offer investors diversification benefits, since longevity risk -- like many insurance risks -- has low correlation with traditional assets \citep{kessler2021new}.

Despite these theoretical advantages, the capital market for long‑dated longevity risk has seen limited success. While markets for other insurance risks, such as catastrophe (CAT) bonds, have matured \citep{cummins2008cat}, long‑horizon longevity transactions remain rare. A notable example is the European Investment Bank’s 25‑year longevity bond, withdrawn in 2005 due to weak demand \citep{blake2006longevity}. By contrast, mortality CAT bonds (e.g., Swiss Re’s Vita Capital series) and more recent longevity sidecars \citep{bugler2021reinsurance} have been issued with shorter three‑to‑five‑year maturities.\footnote{See \url{https://www.artemis.bm/deal-directory/} for a comprehensive deal directory of catastrophe bonds and insurance-linked securities transactions.}

These market patterns are consistent with differences in effective risk aversion across channels. Existing literature links risk aversion to diversification capacity, familiarity with risk, and regulatory constraints \citep{khaw2021cognitive,laeven2009bank,guiso2008risk}. In this paper, we refer to the institute seeking to hedge their longevity risk exposure as the \emph{buyer}, and the institute who provides longevity protection as the \emph{seller}.
In the reinsurance market, buyers -- pension plans, small/medium life insurers, and buy-out firms -- are typically more risk‑averse than the sellers, who are global reinsurers. Reinsurers operate broader, diversified portfolios (life, P\&C, etc.), so longevity risk is a moderate share of their overall exposure. Buyers, by contrast, manage smaller annuitant pools with limited diversification, which raises their effective risk aversion. 
Furthermore, decision‑makers of pension plans and insurers may display aversion to downside outcomes that threaten institutional solvency or job security \citep{ge2021role}. Consequently, reinsurance sellers tend to be less risk‑averse than buyers and are willing to supply long‑dated, higher‑coverage contracts, which buyers accept at higher premiums.

By contrast, in capital‑market settings, sellers, typically investment banks and other institutional investors, are \emph{more} risk‑averse toward longevity exposure due to limited domain expertise and weaker hedging/diversification options than global reinsurers (who often act as buyers). The limited understanding of the risk, combined with the lack of liquidity in the market, makes it difficult for these investors to transfer, diversify, or hedge the assumed longevity risks than the buyers. 

Our analysis is based on a theoretical model under the {\textit{Stackelberg sequential pricing} framework (hereafter, Stackelberg game). In the first step (the buyer’s problem), the utility-maximizing buyer, holding a life-annuity portfolio, chooses an optimal hedge ratio in response to a given risk loading (price). In the second step (the seller’s problem), anticipating the buyer’s choice, the seller sets the risk loading that maximizes her utility.} The main analysis focuses on \emph{indemnity} longevity swaps \citep{dowd2006survivor}, in which the buyer pays fixed amounts and receives floating payments linked to realized survivor counts in her portfolio.  {Extensions to index‑based longevity swap, with which both the fixed and the floating payments are linked to mortality experience of a (broader) reference population, are discussed in Section~\ref{sec:indexbase}.

The Stackelberg formulation has three advantages for our question. First, it reflects the market roles: the seller designs/prices, the buyer accepts and hedges. Second, it yields an \emph{endogenous} demand for risk transfer from the buyer, determined by her risk aversion and the quoted loading. Third, it allows a joint examination of both parties’ {gain in their objective functions relative to no trade (referred to as ``objective gains'' in the sequel)}, recognizing that higher prices reduce the buyer’s demand and hence the seller’s feasible surplus.\footnote{{We note that this Stackelberg framework is a stylized benchmark rather than a literal description of prevailing market microstructure. In practice, longevity risk transfer is typically negotiated and intermediated, and liability-transfer transactions may involve competitive bidding. Our aim is therefore to isolate how contract horizon interacts with heterogeneous risk aversion within a tractable sequential-pricing benchmark.}} 

We apply the dynamic mean–variance (MV) criterion as a performance measure for both buyers and sellers. Since the parties involved in longevity transfers are institutions, the mean–variance criterion is a more appropriate proxy than standard von Neumann–Morgenstern expected utility, which is generally used for individual decision‑making.  {Furthermore, the MV criterion can be interpreted as a simple approximation of a general utility function based on the Arrow–Pratt approximation \citep[see, for example,][]{maccheroni2013alpha}.}\footnote{ {We acknowledge that comprehensively modeling a firm's risk aversion is challenging, as it can be influenced by numerous endogenous and exogenous factors. This paper focuses on analyzing the contract preferences of the two parties and their realized levels of risk aversion, while the determinants of each party's risk aversion lie beyond the scope of this paper.}} Given the time‑inconsistency issue with the dynamic mean–variance criterion, this paper also contributes to the literature on time‑inconsistent dynamic optimization. A simplistic approach is to ignore time inconsistency and optimize period‑by‑period, which is myopic. \citet{basak2010dynamic} instead conceptualize time inconsistency as a strategic game among an individual’s future selves and seek a subgame‑perfect Nash equilibrium -- an approach widely adopted in subsequent work \citep[see, among others,][]{bjork2014mean,bjork2017time,dai2021dynamic}. We adopt this equilibrium approach in our dynamic contracting analysis.

We compare two contract types: a \textit{static} swap that fixes both the hedge ratio and the fixed‑leg payments at inception, and a \textit{dynamic} swap that updates them each period using the latest information. Static swaps provide long‑horizon protection with known hedging costs similar to long‑term reinsurance and could provide higher coverage, whereas dynamic swaps resemble rolling short‑dated instruments (e.g., mortality CAT bonds or longevity sidecars) with more variable costs and potentially lower effective coverage. In numerical analyses, we find that the preferred structure depends on which side is more risk‑averse: \textit{A more risk-averse buyer favors the static contract, as it better aligns with her long-term hedging needs, despite the higher premiums. Conversely, a more risk-averse seller prefers selling the dynamic contract, even at a lower price, to avoid long-term commitments.}  {When index-based swaps are used instead of indemnity swaps, we find that the equilibrium risk loadings, as well as the {objective} gains for both parties, are lower than in the case of indemnity swaps. However, the relationship between risk aversions of the two counterparties and the preference of contract types remains unchanged.} These findings help rationalize observed market outcomes and guide contract design. In particular, while the development of the long-term capital market is admittedly
challenging, our model predictions suggest that dynamic contracts, such as mortality CAT bonds or longevity sidecars, represent a promising initial step toward developing an active longevity-linked capital market.

As a robustness check and model extension, we also consider the case where the more risk-averse seller faces model ambiguity about the mortality distribution of the hedged portfolio. Using the max--min expected-utility framework \citep{GS89}, the seller evaluates outcomes under a set of candidate survival distributions and chooses the loading against the worst-case prior. We find that dynamic-contract preference persists for more risk-averse sellers and that even moderate ambiguity can eliminate the market for static contracts, providing further support for the scarcity of long-term capital-market longevity transfers.

{In practice, longevity exposure moves through three closely related routes. (i) \emph{Liability transfer} (buy-out): an insurer (often a specialist buy-out firm) assumes the legal obligation and pays members directly; (ii) \emph{Bulk-annuity asset} (buy-in): the plan purchases a bulk annuity and is reimbursed for benefits while \emph{retaining} the legal liability; (iii) \emph{Longevity-only hedges}: only the longevity risk in the liability is transferred, most commonly via longevity swaps generate payments exclusively based on the mortality experience of a portfolio/population. Our problem belongs to (iii). Importantly, buy-out firms currently drive much of the market’s swap demand by acting as pass-throughs: after buy-ins/buy-outs, they typically lay off a large share of the longevity leg via longevity swaps with reinsurers. Although buy-outs package other risks of the liability, including investment and credit risks, alongside longevity risk, the longevity leg is frequently carved out and hedged with instruments matching our setup \citep[see, for example,][]{d2018risking}. Hence our contract-type comparison speaks directly to the design choice faced by buy-out firms and other hedgers.}

Finally, it is worth noting that, while static contracts have found
widespread use in both the literature and the real-world insurance industry,
dynamic contracts are becoming increasingly relevant in insurance research %
\citep[see, for example,][]{wong2017managing,wang2023time}.
Nonetheless, the full range of advantages and disadvantages of these
contract types for buyers and sellers remains largely unexplored.   {Motivated by this gap, the
goal of this paper is to study the preferences for static versus dynamic
contracts by analyzing the {mean-variance} objective of both parties and establishing their
connection to risk attitudes in the context of longevity risk transfer.} We
posit that the patterns observed in managing longevity risk -- specifically,
the preference for static contracts by less risk-averse sellers and for
dynamic contracts by more risk-averse sellers -- may apply to other economic
contracting scenarios as well.

Our work contributes to the literature on longevity‑linked markets and contract design \citep{blake2013new,biffis2016cost,chen2022collective,chen2022tail,borger2023economics,blake2023longevity,chen2025learning}, survivor derivatives and pricing \citep{dowd2006survivor,dawson2010survivor}, dynamic longevity hedging \citep{wong2017managing,li2018dynamic,wang2023time,chen2024coping}, and mortality modeling \citep{boonen2017modeling,li2019forecast,li2021forecasting}. It also relates to ambiguity and robust hedging of longevity risk \citep{cairns2013robust,li2017robust} and to a separate but related strand of literature on information asymmetry in longevity transactions \citep{biffis2010securitizing,biffis2013informed,biffis2014keeping,chen2023optimal}. Methodologically, we build on sequential leader–follower pricing and time-inconsistent dynamic optimization \citep{basak2010dynamic,bjork2010general,bjork2014mean,bjork2014theory,bjork2017time,dai2021dynamic}.

The remainder of the paper is organized as follows. Section~\ref{sec:model} introduces the notation and the Stackelberg framework. Section~\ref{sec:numerical_analysis} presents numerical analyses of contract preferences across risk‑aversion regimes. Section~\ref{sec:information_asymmetry} incorporates seller‑side {model ambiguity in the capital-market setting.} Section~\ref{sec:indexbase} extends the analysis to index‑based swaps. Section~\ref{sec:conclusion} concludes. Technical details and supporting results are gathered in the Appendix.

\section{Model Setting} \label{sec:model}

This section introduces the model and notation (Section~\ref{sec:notation}), defines the two contract types (Section~\ref{sec:contract}), and formulates the Stackelberg game (Section~\ref{S:Game}).

\subsection{Surplus Dynamics}

\label{sec:notation}

We commence by setting the groundwork for a contract arrangement concerning
an indemnity longevity swap, involving two participating parties: the seller
and the buyer. We assume, for ease of understanding, that the buyer
possesses a life annuity portfolio consisting of a single cohort of
policyholders aged $x$ at time $0$. The buyer makes annual unitary payments to the policyholders.
It is also assumed that
these policyholders are homogeneous, sharing identical future survival
probabilities. The buyer's objective is to either
eliminate or reduce this risk through longevity swaps\footnote{%
According to Artemis, longevity swaps and longevity risk reinsurance
represent the most common form of longevity risk transfer in the current
market. 
In terms of cash flow structure,
longevity swaps have zero net cost at issuance (discussed later in this section), whereas longevity
reinsurance typically require an upfront premium payment. The inclusion of
longevity reinsurance within our analysis can be achieved straightforwardly.
Source: \url{https://www.artemis.bm/library/what-is-longevity-risk-transfer/}%
.}.

Throughout this paper, we concentrate exclusively on longevity risk, setting
aside financial risk. We work within a filtered probability space 
$(\Omega,\mathcal{F},\mathbb{F},\mathbb{P})$ and a discrete-time framework with
annual decision dates $\boldsymbol{\tau}=\{0,1,\ldots,T\}$, where $T$ denotes the
maximum survival horizon of the policyholders. The filtration 
$\mathbb{F}=\{\mathcal{F}_t\}_{t\in\boldsymbol{\tau}}$ represents the information
available up to and including time $t$. All random variables and stochastic
processes relevant to our study are defined on this space and are
$\mathbb{F}$-adapted. Key notation used in our analysis includes:
\begin{itemize}
\item $\mathbb{P}$: The  {reference measure} for the expected count
of policyholders alive during contract formulation. $\mathbb{P}$ is used for pricing the longevity swaps.

\item $_{t}\mathrm{p}_{x}$: The  {real-world} $t$-year survival probability 
for the policyholders aged $x$ at time 0.

Our focus is on a single cohort aged $x$ at time 0, and thus survival
probabilities discussed in this paper pertain solely to this cohort. Throughout this paper, $_{t-s}\mathrm{p}_{x+s}$ represents the $(t-s)$-year
survival probability for policyholders aged $x+s$ at time $s$ for any $%
s,t\in \boldsymbol{\tau }$ such that $s+t\leq T$ and $s<t$.  {For an individual aged \(x+s\) at \(t=s\), $_{t-s}\hat{\mathrm{p}}_{x+s}$ denotes the estimated $(t-s)$-year
survival probability under \(\mathbb{P}\).} 

\item $l_{0}$: Initial count of policyholders aged $x$ at time 0 within the
annuity portfolio.

\item $l_{t}$: The random count of survived policyholders in year $t$. Note that $l_{t}$ is $\mathcal{F}_{t}$-measurable and remains unobserved before $t$. Given information up to any $s$, where $s<t$, we assume that $l_{t}$ follows a binomial distribution,  {$l_{t}|\mathcal{F}_{s}\sim B(l_{s},\,_{t-s}\mathrm{p}_{x+s})$.} 

\item $\hat{l}_{t}$: The expected count of policyholders alive in year $t$
under the reference measure $\mathbb{P}$. This expected value
takes a different form in static contracts and dynamic contracts (discussed in Section~\ref{sec:contract}),
depending on whether the information of the realized count of survived
policyholders is used.  {In our analysis, we particularly focus on the \emph{unconditional} expectation $\mathbb E^{\mathbb P}[l_t]$ (as seen from time $0$) and the \emph{conditional} one-step-ahead expectation $\mathbb E^{\mathbb P}[l_t\mid\mathcal F_{t-1}]$ (as seen from time $t-1$). 
When we later specify contract timing, these two expectations will correspond to two fixing conventions for the fixed leg. 
Formally, we adopt the temporary shorthand
\[
\hat l_t^{[0]}:=\mathbb E^{\mathbb P}[l_t], 
\qquad 
\hat l_t^{[t-1]}:=\mathbb E^{\mathbb P}[l_t\mid\mathcal F_{t-1}].
\]
For simplicity of notation, we will omit the superscripts $[0]$ and $[t-1]$ when no confusion arises.
}
\end{itemize}

 {We assume that, with complete information, the estimated survival probabilities under the reference measure coincide with the real-world survival probabilities: $_{t-s}\hat{\mathrm{p}}_{x+s}= {_{t-s}\mathrm{p}_{x+s}}$ for all $s$ and \(t\). However, in the presence of {model ambiguity} (discussed in Section~\ref{sec:information_asymmetry}), the real-world survival probabilities are unknown by the seller. In this case, the reference measure $\mathbb{P}$ serves as the most plausible proxy for the real-world measure.}

The surplus of the buyer and the seller at time $t$ is denoted by $B_t$ and $%
S_t$, respectively. We assume the buyer's initial surplus, $B_0$, is known.
It could represent, for example, the total premium income from annuities. Our analysis exclusively
focuses on longevity risk, assuming a financial market with a single,
risk-free asset and a constant interest rate $r$. The buyer's unhedged surplus evolves as
\begin{align}
B_t = B_{t-1}(1+r) - l_t,  \label{eq:X_nohedge}
\end{align}
i.e., the prior surplus accrues at rate $r$ and is reduced by annuity payments at time $t$. Longevity risk originates from the uncertainty in future survivor counts, $l_t$.

In the baseline analysis, we explore longevity risk transfer through indemnity longevity swaps -- agreements to exchange future cash flows based on the realized counts of surviving policyholders within the buyer's portfolio. With an indemnity longevity swap, the seller offers partial or complete protection against the
uncertainty in $l_{t}$ in exchange for a series of premiums.
{The longevity swap is a $T$-year contract, initiated at \emph{zero net cost} at \(t=0\), and cash-settled via periodic net payments.}\footnote{ {In practice, collateral/initial margin may apply, which induces $t=0$ costs. We abstract from these funding frictions in our analysis.}}
In year $t$, the seller pays a floating
amount $u_{t-1}l_{t}$ to the buyer and, in return, receives a fixed
amount $u_{t-1} {(1+\tilde{\eta}_t)}\hat{l}_{t}$, where $u_{t-1}\in [0,1]$ is the 
\textit{hedge ratio} at time $t$, $ {\tilde{\eta}_t}$ is the time-$t$ longevity risk loading, and
$\hat{l}_{t}$ is the expected count of policyholders alive in
year $t$ under the reference measure $\mathbb{P}$.\footnote{ {Formally, $\hat{l}_{t}$ in the fixed payments should be denoted by $\hat{l}_{t}^{[\tau]}$, where $\tau\in\{0,t-1\}$ represents the measurability of $\hat{l}_{t}$. Contract terms will be formally defined in Section~\ref{sec:contract}.}} 
In each year $t$, if $l_{t}> {(1+\tilde{\eta}_t)}\hat{l}_{t}$, i.e., the realized number of policyholders
alive is larger than the expected value plus the risk premium, the
buyer receives a net payment $u_{t-1}\bigl(l_{t}- {(1+\tilde{\eta}_t)}\hat{l}_{t}\bigr)$; 
if $l_{t}< {(1+\tilde{\eta}_t)}\hat{l}_{t}$, she makes a net
payment $u_{t-1}\bigl( {(1+\tilde{\eta}_t)}\hat{l}_{t}-l_{t}\bigr)$ to the seller. Extensions to index‑based longevity swap are discussed in Section~\ref{sec:indexbase}.

Incorporating the longevity swap modifies the buyer's hedged
surplus, $\{B_{t}(u_{t-1})\}_{t=1,\ldots,T}$, to
\begin{align}
B_{t} \;=\; B_{t-1}(1+r) - l_{t} + u_{t-1}\bigl(l_t- {(1+\tilde{\eta}_t)}\hat{l}_{t}\bigr).
\label{eq:X}
\end{align}
For the seller, assuming no preexisting assets or liabilities, the surplus evolves as
\begin{align}  \label{eq:Y}
S_{t} \;=\; S_{t-1}(1+r) - u_{t-1}\bigl(l_t- {(1+\tilde{\eta}_t)}\hat{l}_{t}\bigr).
\end{align}

\subsection{Static and Dynamic Contracts} \label{sec:contract}

We consider two types of longevity swap contracts, distinguished by how the hedge ratio and fixed payments are determined:
\begin{enumerate}
    \item \textbf{Static Contract}: The hedge ratio is \textit{constant}, $u_{t-1} \equiv u$, for all $t=1,\ldots,T$. In addition, $\hat{l}_{t}$ in the fixed payment for each $t$ is determined at the \textit{beginning of the contract} (time 0), i.e., for $t=1,\ldots,T$,
\begin{align}
\hat{l}_{t} {^{[0]}} = \mathbb{E}^{\mathbb{P}}\!\left[ l_{t}\right]={}_{t}\hat{\mathrm{p}}_{x}\, l_{0}. \label{l hat static}
\end{align}
    \item \textbf{Dynamic Contract}: The hedge ratio is allowed to \textit{change over time}. Specifically, $u_{t-1}$ is chosen at time $t-1$ and is $\mathcal{F}_{t-1}$-measurable. The fixed-leg forecast $\hat{l}_{t}$ is also updated based on $\mathcal{F}_{t-1}$ at each $t-1$: 
    \begin{align}
    \hat{l}_{t} {^{[t-1]}} = \mathbb{E}^{\mathbb{P}}\!\left[l_{t}\mid\mathcal{F}_{t-1}\right] = \hat{\mathrm{p}}_{x+t-1}\,l_{t-1}, \label{l hat dynamic}
    \end{align}
    where $\hat{\mathrm{p}}_{x+t-1}$ is the estimated one-year survival probability for a policyholder aged $x+t-1$ in year $t-1$, and $l_{t-1}$ is the observed number of policyholders alive at time $t-1$.
\end{enumerate}
 {For model simplicity, we assume the risk loading 
\begin{equation}
    \tilde{\eta}_t=(1+\delta_L)^{t}\eta \label{eq:loading}
\end{equation} 
for all $t$ and both contract types, where $\delta_L$ is a predetermined positive value that ensures the risk loading is larger in later years of the contract. Given the specification in Equation~\eqref{eq:loading}, selecting the $T$-dimensional risk loading $\tilde{\eta}_{t\in\{1,...,T\}}$ simplifies to choosing a single parameter $\eta$.}

It is worth noting that both contract types share the same maturity date, $T$%
. The terms `static' and `dynamic' are used to describe the structural
difference in the scope of risk coverage. The constant hedge ratio and predetermined fixed payments make
hedging costs of static contracts completely known from the outset. This provides long-term, \textit{%
high-risk coverage} throughout
the contract duration. 
 {In contrast, since the hedge ratio and fixed payment are periodically updated, dynamic contracts primarily hedge next year’s cash flow rather than the entire remaining horizon.} In this sense, like a
series of one-year forwards, dynamic contracts offer \textit{low-risk
coverage}. With these contract specifications in place, we next formulate the Stackelberg game that determines the equilibrium hedge ratios and risk loadings.

\subsection{Stackelberg {Sequential Pricing} Framework}
\label{S:Game}

To identify optimal contracts, we adopt a {Stackelberg sequential pricing
framework (hereafter, Stackelberg game)}, recognizing the seller as the price maker and the buyer as the
price taker. Specifically, the seller, acting as the {leader}, offers
a range of swap contracts to the buyer, each distinguished by a specific
longevity risk loading $\eta $. The buyer responds by
selecting an optimal hedge ratio -- either constant or time-varying
depending on the contract type -- for each given $\eta $. With knowledge of
the buyer's optimal choice of hedging ratio, the seller then determines the
optimal value of risk loading $\eta $ that best serves her own objective. {In this paper, we abstract our analysis from hidden types, incentive constraints, or screening.}
To facilitate comparison, we assume that the seller presents only one type of contract at a time.
Therefore, the optimization problem for static and dynamic contracts is
conducted through two separate Stackelberg game analyses.

Our choice to employ the Stackelberg game framework is based on several reasons. 
First, the Stackelberg game framework can be seen as a middle ground between
the classical monopoly and competitive models \citep{rothschild1976equilibrium, stiglitz1977monopoly}, which are two extremes. It is
widely acknowledged that real insurance markets lie between perfect competition and monopoly. Even though the seller is the leader of the game, it is clearly not optimal to set the risk loading
extremely high as the buyer's demand would diminish. Second, the buyer's
demand is \textit{endogenously} determined according to the seller's pricing strategy. Due to the lack of empirical data, it is generally difficult to accurately estimate the buyers' demand curve in response to changes in insurance prices. In the Stackelberg game, buyers' demand is well incorporated by solving for the buyer's optimal hedge ratio in the first step. 
Furthermore, unlike existing works, we focus on investigating the relationship between contract type preferences (static or dynamic) and the risk preferences of the involved parties.

In the Stackelberg game, we select the mean-variance (MV) criterion as the objective function for both players. This choice is based on the following reasons. First, stemming from Markowitz's modern portfolio theory, the MV criterion has been widely used as a common performance measure in asset management for firms, applicable to both single-period and multi-period problems \citep[see, for example,][]{ait2001variable,acharya2005asset,basak2010dynamic,bjork2014mean,dai2021dynamic}. Second, the MV criterion can be viewed as an approximation to the standard von Neumann–Morgenstern expected utility maximization, as per the Arrow–Pratt approximation of certainty equivalence \citep{maccheroni2013alpha}. This implies that the findings based on the MV criterion should remain similar when general utility functions are used. Third, the MV criterion simply reduces to the risk-neutral case when the risk aversion parameter becomes zero. 

Formally, for the \textit{static contract}, the buyer selects the optimal constant hedging ratio $u$ given a risk loading $\eta\geq 0$. This is achieved by solving the following optimization problem:
\begin{align}\label{insurer objective_static}
     \sup_{u}\left\{\mathbb{E}^{\mathbb{P}}\left[ B_{T}(u,\eta)\right] -\frac{\gamma_b}{2}\mathrm{Var}^{\mathbb{P}}\left[B_{T}(u,\eta)\right]\right\},
\end{align}
where $\gamma_b$ represents the buyer's degree of risk aversion, and $B_{T}(u,\eta)$ denotes the buyer's surplus at time $T$, given hedge ratio $u$ and risk loading $\eta$. The seller, considering the buyer's optimal response, $u^{\ast}(\eta)$, then selects the optimal risk loading $\eta$ to maximize the mean-variance criterion:
\begin{equation}
\sup_{\eta\geq 0}\left\{\mathbb{E}^{\mathbb{P}}\left[ S_{T}\bigl(u^{\ast}(\eta)\bigr)\right] -\frac{\gamma_s }{2}\mathrm{Var}^{%
\mathbb{P}}\left[S_{T}\bigl(u^{\ast}(\eta)\bigr)\right]\right\},
\label{MEU}
\end{equation}
where $S_{T}\bigl(u^{\ast}(\eta)\bigr)$ is the seller's surplus at time $T$ given hedge ratio $u^{\ast}(\eta)$ and $\gamma_s$ denotes the seller's degree of risk aversion.

In the context of a \textit{dynamic contract}, the buyer adjusts her strategy to maximize the mean-variance criterion with a \textit{time-varying} hedge ratio:
\begin{align}\label{insurer objective}
     \sup_{\textbf{u}}\left\{\mathbb{E}^{\mathbb{P}}\left[ B_{T}(\textbf{u},\eta)\right] -\frac{\gamma_b}{2}\mathrm{Var}^{\mathbb{P}}\left[B_{T}(\textbf{u},\eta)\right]\right\},
\end{align}
with $\textbf{u}=(u_t)_{t\in\{0,1,2,...,T-1\}}$. The seller's corresponding optimization problem is then defined as:
\begin{equation}
\sup_{\eta\geq 0}\left\{\mathbb{E}^{\mathbb{P}}\left[ S_{T}\bigl(\textbf{u}^{\ast}(\eta)\bigr)\right] -\frac{\gamma_s }{2}\mathrm{Var}^{%
\mathbb{P}}\left[S_{T}\bigl(\textbf{u}^{\ast}(\eta)\bigr)\right]\right\},
\label{MEU_dynamic}
\end{equation}
with $\textbf{u}^{\ast} = \bigl(u_t^{\ast}(\eta)\bigr)_{t\in\{0,1,2,...,T-1\}}$. 

We would like to make two remarks regarding the proposed Stackelberg game. First, throughout the formulation of the Stackelberg game, we assume that the buyer and the seller exhibit risk aversion $\gamma_{b} > 0$ and $\gamma_{s} \geq 0$, respectively. The buyer's risk aversion is crucial because it determines her willingness to engage in longevity risk transfer. If the buyer were not risk-averse, she would lack the motivation to transfer any longevity risk considering the risk loading imposed by the seller in addition to the expected loss. This additional cost is a deterrent unless mitigated by the benefits of risk transfer. In contrast, the seller could be either risk averse or risk neutral. Throughout the analysis, we consider a general situation of a risk-averse seller. As a specific scenario, we will discuss the results with a risk-neutral seller ($\gamma_{s} = 0$) in Appendix~\ref{sec:risk_neutral_seller}. Second, in this Stackelberg setup, {no trade is always feasible in this benchmark. The buyer can choose not to hedge, corresponding to $u_{t-1}=0$ for all $t$, and the seller can refrain from offering an attractive loading. Positive-trade outcomes are therefore interpreted relative to each party's own no-trade benchmark. Otherwise, the equilibrium outcome is simply no trade.}

Finally, note that time inconsistency arises only in the buyer's problem when a dynamic
contract is considered. In other words, there is no such issue for the
seller's problem or the buyer's problem under a static contract, since the
control variables in these cases remain constant over time. For the buyer's problem with dynamic contract, a more sophisticated planning approach is required. Specifically, the buyer shall treat the problem as a
game in which the \textquotedblleft players\textquotedblright\ are the buyer
future selves, each making decision sequentially over time. The optimal
strategy, referred to as a sub-game perfect equilibrium, is then determined
by solving for the hedge ratio backward in time. {Furthermore, because the dynamic mean-variance problem is time-inconsistent, the dynamic contract should be interpreted as a tractable benchmark under asymmetric commitment: the seller commits ex ante to the loading rule, whereas the buyer solves the dynamic hedging problem via the standard intrapersonal equilibrium approach. Recent work by \cite{bayraktar2025time} argues more generally that equilibrium notions in time-inconsistent environments require careful interpretation and that stronger consistency notions may be desirable in richer settings. Developing such consistency concepts in our sequential pricing game is an interesting future research direction.}

\section{Numerical Analysis}
\label{sec:numerical_analysis}

This section evaluates the implications of static versus dynamic contracts {to the mean-variance objectives of both parties}. The base scenario parameters are an interest rate ($r$) of 0.02, an initial age ($x$) of 60, an initial annuitant count ($l_0$) of 100,000, a contract duration ($T$) of 40 years, and $\delta_L=0.03$ in the risk loading formula \eqref{eq:loading}. Without loss of generality, the initial wealth of both the buyer ($B_0$) and the seller ($S_0$) is assumed to be zero.  We consider two cases in which the buyer or the seller is more risk-averse, respectively. Results with a risk-neutral seller ($\gamma_{s} = 0$) are discussed in Appendix~\ref{sec:risk_neutral_seller}. 

\subsection{Simulating Future Survival Probabilities} \label{sec:mortality}

We use the Age-Period-Cohort-Improvement (APCI) mortality model to simulate future survival probabilities. The APCI model is used by the CMI Mortality Projections Committee\footnote{Source: \cite{cmi2016cmi}.} to generate life tables for U.K. life insurers and pension plans. It models the annual death rates at age $x$ and year $h$, $m_{x,h}$, logarithmically:
\begin{align}
\ln(m_{x,h})=\beta_x^{(1)}+\beta_x^{(2)}(h-\bar{h})+ \kappa_h+\theta_c+\sigma_{\omega_x}\omega_{x,h}, \label{eq:apci}
\end{align}
where $\beta_x^{(1)}$ and $\beta_x^{(2)}$ represent the age and period effects, respectively, $\bar{h}$ is a normalization parameter, $\kappa_h$ is a time-varying mortality improvement factor, $\theta_c$ denotes the cohort effect, $\sigma_{\omega_x}$ is the standard deviation of error terms for age $x$, and $\omega_{x,h}$ are independent standard normal distributed error terms.

Forecasting future survival probabilities involves projecting the mortality improvement factor, $\kappa_h$. We model $\kappa_h$ as a random walk following an ARIMA(0,1,0) process:
\begin{align}
    \kappa_h =\kappa_{h-1}+\sigma_{\kappa}\varepsilon_h, \label{eq:kappa}
\end{align}
where $\sigma_{\kappa}$ is the standard deviation of the error terms, and $\varepsilon_h$ are independent standard normal distributed error terms.

The APCI model is calibrated to U.K. uni-sex death rates sourced from the Human Mortality Database\footnote{Source: \url{https://www.mortality.org/}.}, using one-age, one-year death rates for ages 20 to 100 and years 1956 to 2020. The model parameters in Equation \eqref{eq:apci} are estimated via maximum likelihood \citep{richards2019stochastic}. Subsequently, the ARIMA(0,1,0) model \eqref{eq:kappa} is fitted to the estimated mortality improvement parameters. After estimating the model, the estimated future survival probabilities under the  {reference measure} are obtained by standard simulation practice.  {After estimation, we forward-simulate $K$ paths of the mortality-improvement factor via \eqref{eq:kappa} and, for each path, map them through the APCI observation equation \eqref{eq:apci} to obtain simulated death rates $m_{x,h}$. Cohort survival probabilities then follow from the standard exponential link between integrated death rates and survival (e.g., $_{t-s}p_{x+s}\approx \exp\{-\sum_{\mu=0}^{t-s} m_{x+s+\mu,\,s+\mu}\}$), and Monte Carlo averages across the $K$ paths yield the forecasts $_{t-s}\hat p_{x+s}$ under the reference measure.
}

\subsection{{Objective Gains Relative to No Trade}}
\label{sec:reinsurance_market}

We first consider the scenario where the buyer is more risk-averse, with $(\gamma_b,\gamma_s)=(0.05,0.02)$. Figure~\ref{fig:HP_no_ambiguity} reports the {objective} gains for both parties under static and dynamic contracts. For any given $\eta$, a party’s {objective} gain equals the difference between its optimal expected utility from trading at the equilibrium hedge volume for that $\eta$ and its no-trade utility. Hence, a positive value indicates a beneficial transaction at that $\eta$.

From Figure~\ref{fig:HP_no_ambiguity}(a), several observations emerge from the seller’s perspective. First, the range of $\eta$ values yielding a positive {objective}  gain is markedly broader for the static contract than for the dynamic one, indicating that static contracts more readily support transactions. Second, the $\eta$ values consistent with a functioning market are higher under static contracts.  {The optimal $\eta^\ast$ that yields the maximal {objective}  gain for the seller is around 0.06 under the static contract, whereas this value is slightly under 0.02 under the dynamic one.} This reflects the larger longevity risk the seller assumes over the entire horizon. Third, the maximum {objective}  gain under static contracts is substantially higher than under dynamic contracts, because buyers place greater value on the more extensive and predictable risk coverage of static designs and thus accept higher $\eta$ for static than for dynamic contracts. Figure~\ref{fig:HP_no_ambiguity}(b) shows the buyer’s {objective}  gains. The bold dot on each line marks the gain at the seller’s optimal $\eta^\ast$. The equilibrium static contract delivers a considerably larger {objective}  gain for the buyer.

\begin{figure}[htbp!]%
    \centering
    \subfloat[Seller's {objective}  gain]{\includegraphics[width=7.5cm]{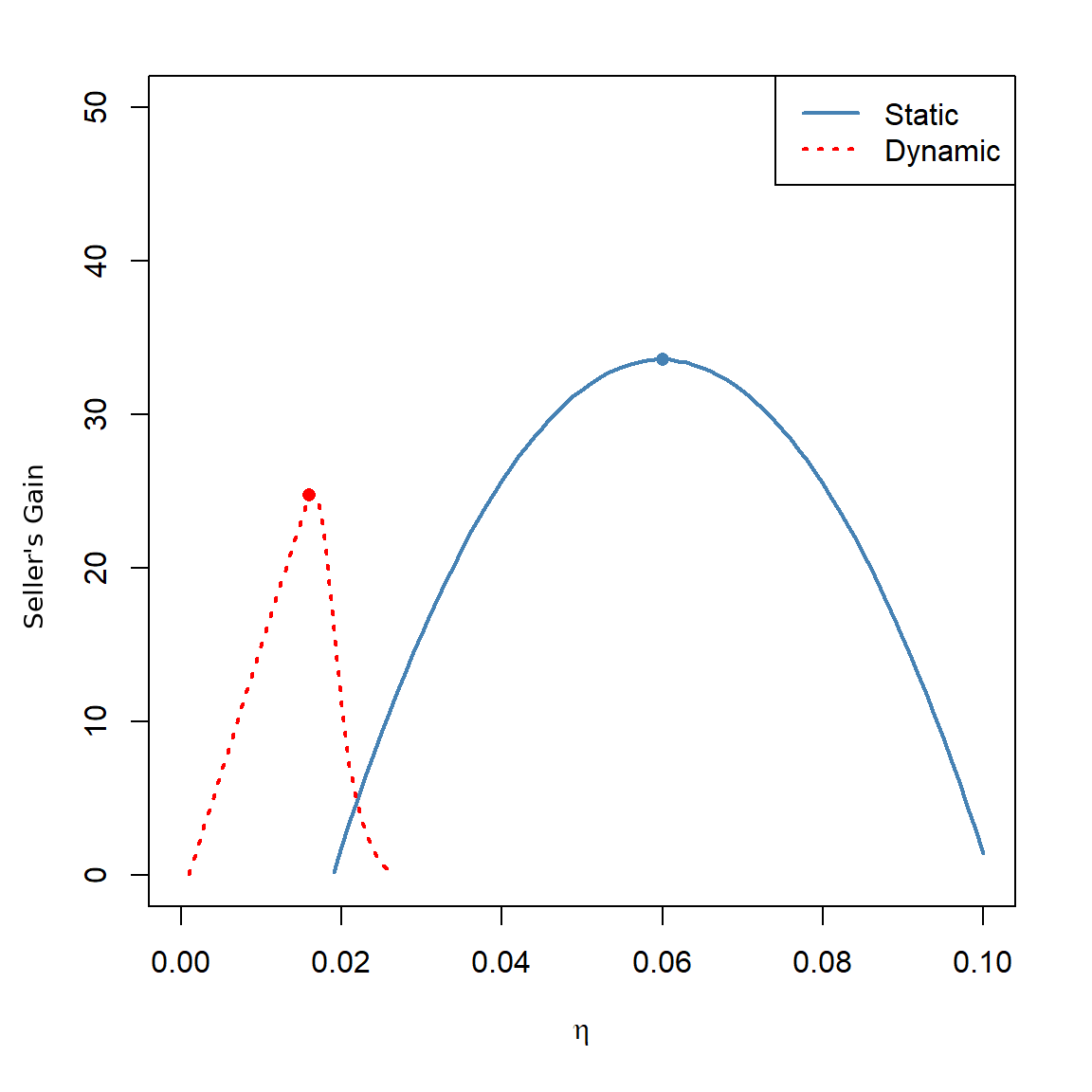} }%
    \qquad
    \subfloat[Buyer's {objective}  gain]{\includegraphics[width=7.5cm]{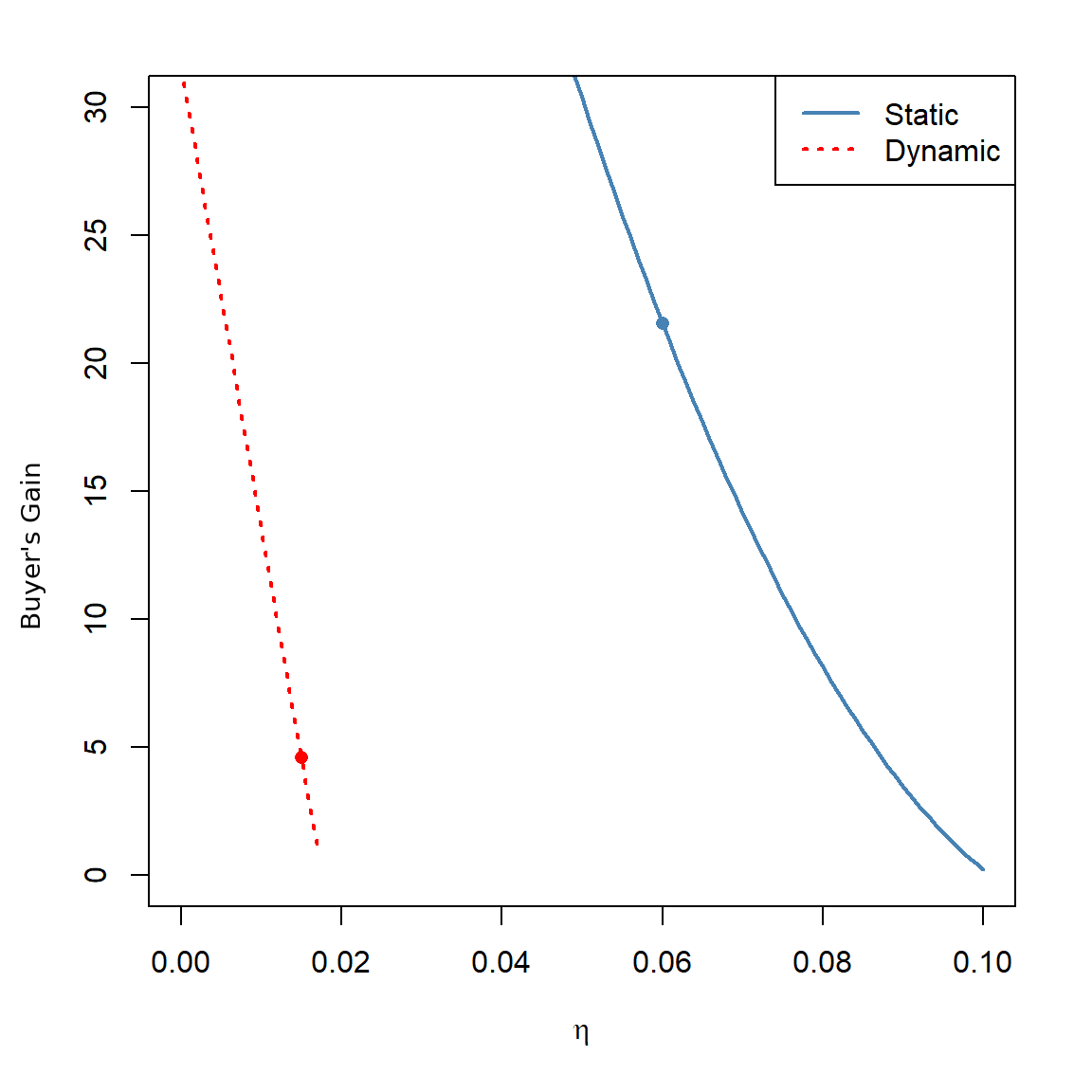} }%
    \caption{{Objective}  gains of the seller (left) and the buyer (right) when trading static vs.\ dynamic contracts  {with $(\gamma_b,\gamma_s)=(0.05,0.02)$}. The bold dot indicates the gains at $\eta^\ast$.}
    \label{fig:HP_no_ambiguity}%
\end{figure}

Next, we examine the case in which the seller is more risk-averse, with $(\gamma_b,\gamma_s)=(0.05,0.2)$. Figure~\ref{fig:HP_t_capital} displays the gains across $\eta$ for each contract type. From the seller’s perspective (Figure~\ref{fig:HP_t_capital}(a)), the $\eta$-region supporting transactions remains broader for the static contract but narrows materially relative to Figure~\ref{fig:HP_no_ambiguity}(a). Importantly, the dynamic contract now delivers a higher maximal gain, indicating a shift in preference toward dynamic designs as the more risk-averse seller avoids committing to long-dated protection. Similarly, Figure~\ref{fig:HP_t_capital}(b) shows that the buyer would also achieve a higher gain with a dynamic contract due to its lower price.

\begin{figure}[htbp]%
    \centering
    \subfloat[Seller's {objective} gain]{\includegraphics[width=7.5cm]{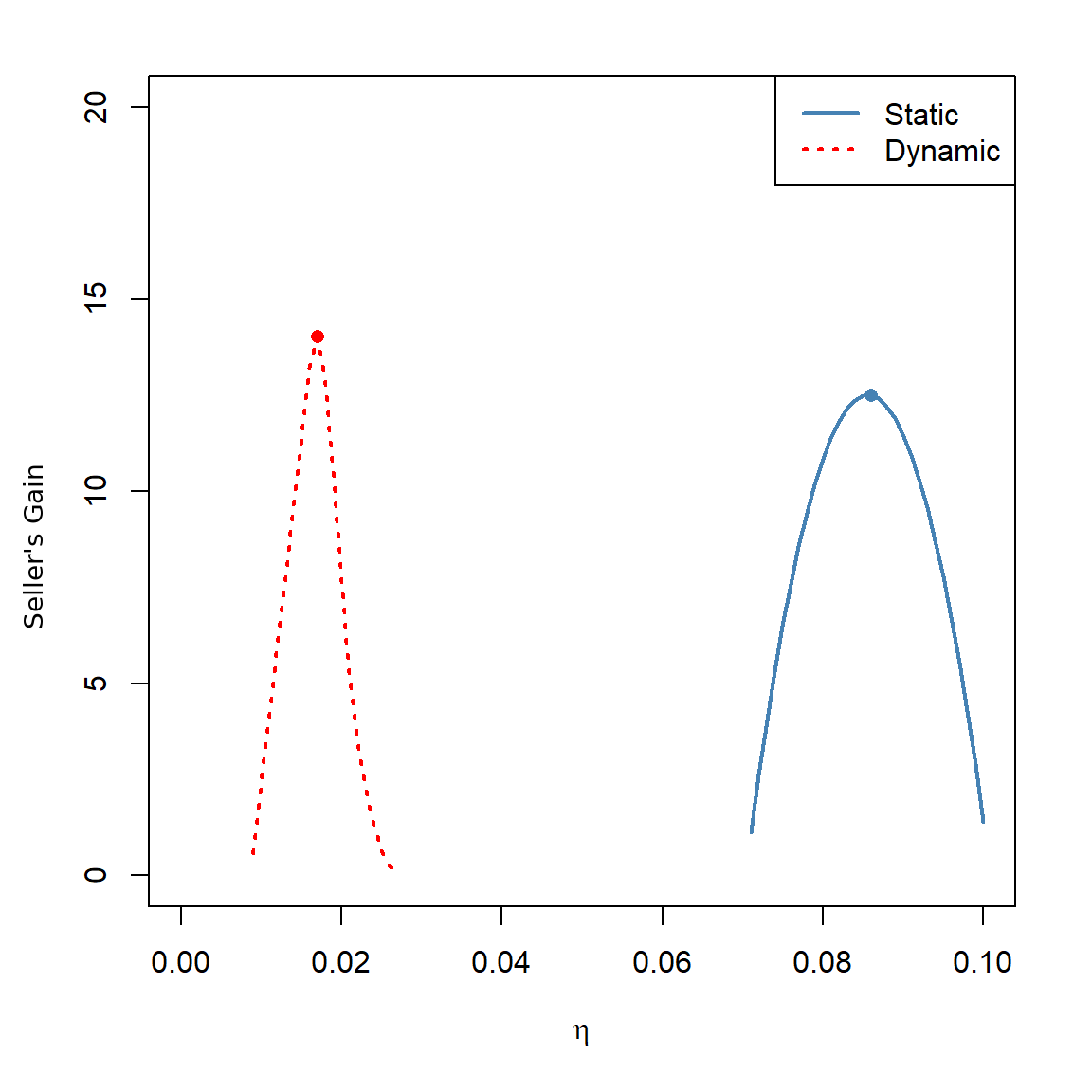} }%
    \qquad
    \subfloat[Buyer's {objective} gain]{\includegraphics[width=7.5cm]{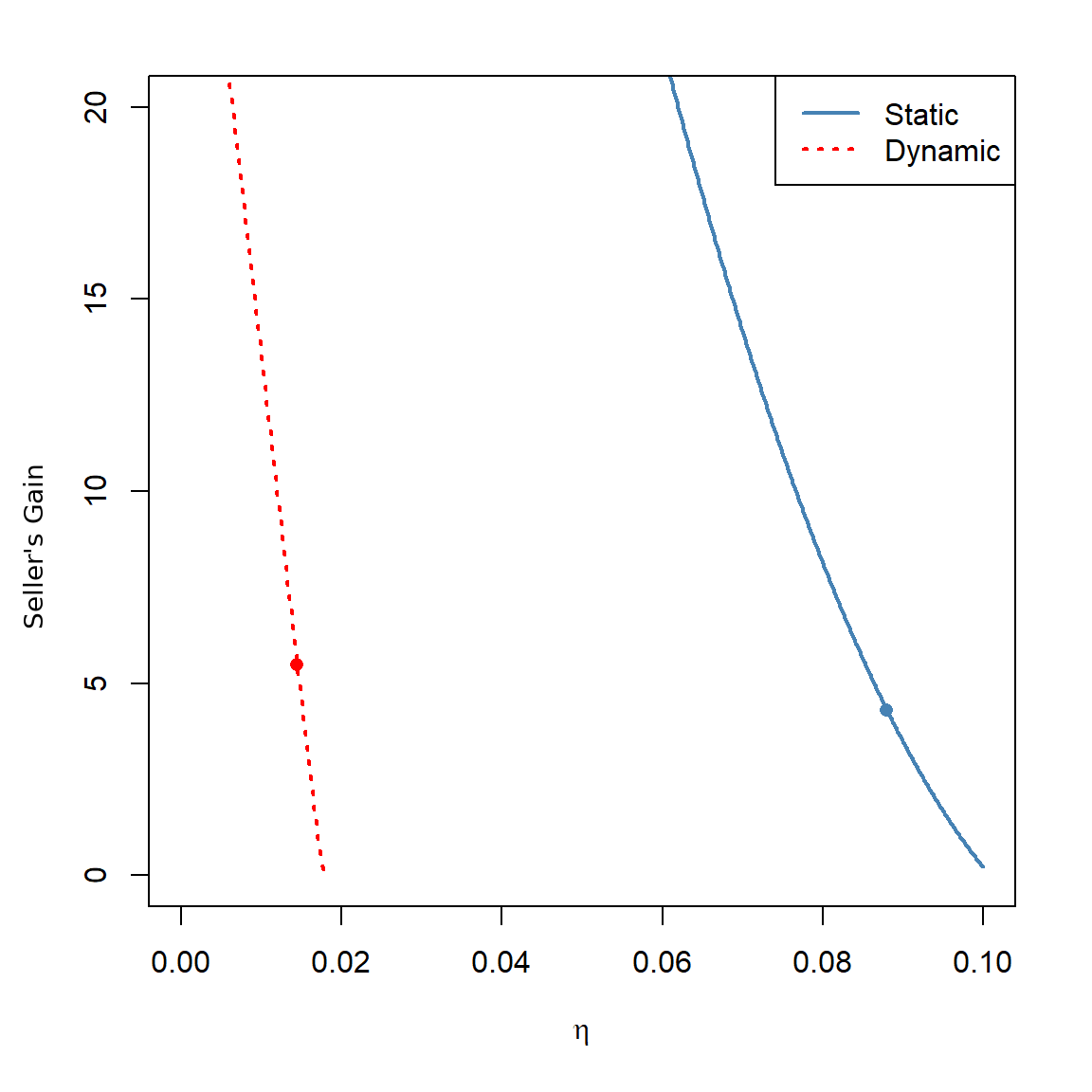} }%
    \caption{{Objective} gains of the seller (left) and the buyer (right) when trading static vs.\ dynamic contracts with $(\gamma_b,\gamma_s)=(0.05,0.2)$.}
    \label{fig:HP_t_capital}%
\end{figure}

\subsection{Fixed Payments and Optimal Hedge Ratios}

To understand risk coverage under each design, we examine the time profiles of fixed payments and the buyer’s optimal hedge ratios. For concreteness, we discuss the more risk-averse buyer case. Results with a more risk-averse seller are qualitatively similar.

Figure~\ref{fig:HPPAY_t}(a) shows the static contract’s fixed payments (in thousands) under a full hedge ($u=1$). These payments decline over time with the expected survivor count. Figure~\ref{fig:HPPAY_t}(b) shows the dynamic contract’s fixed payments, where each year-$t$ payment is based on the latest observed survivor count, $l_{t-1}$. Because $l_{t-1}$ is random, future fixed payments are uncertain. The figure reports 95\% confidence intervals to reflect this variability. Relative to the static case, the dynamic fixed leg is less predictable ex ante.

\begin{figure}[htbp]%
    \centering
    \subfloat[Static contract]{\includegraphics[width=7.5cm]{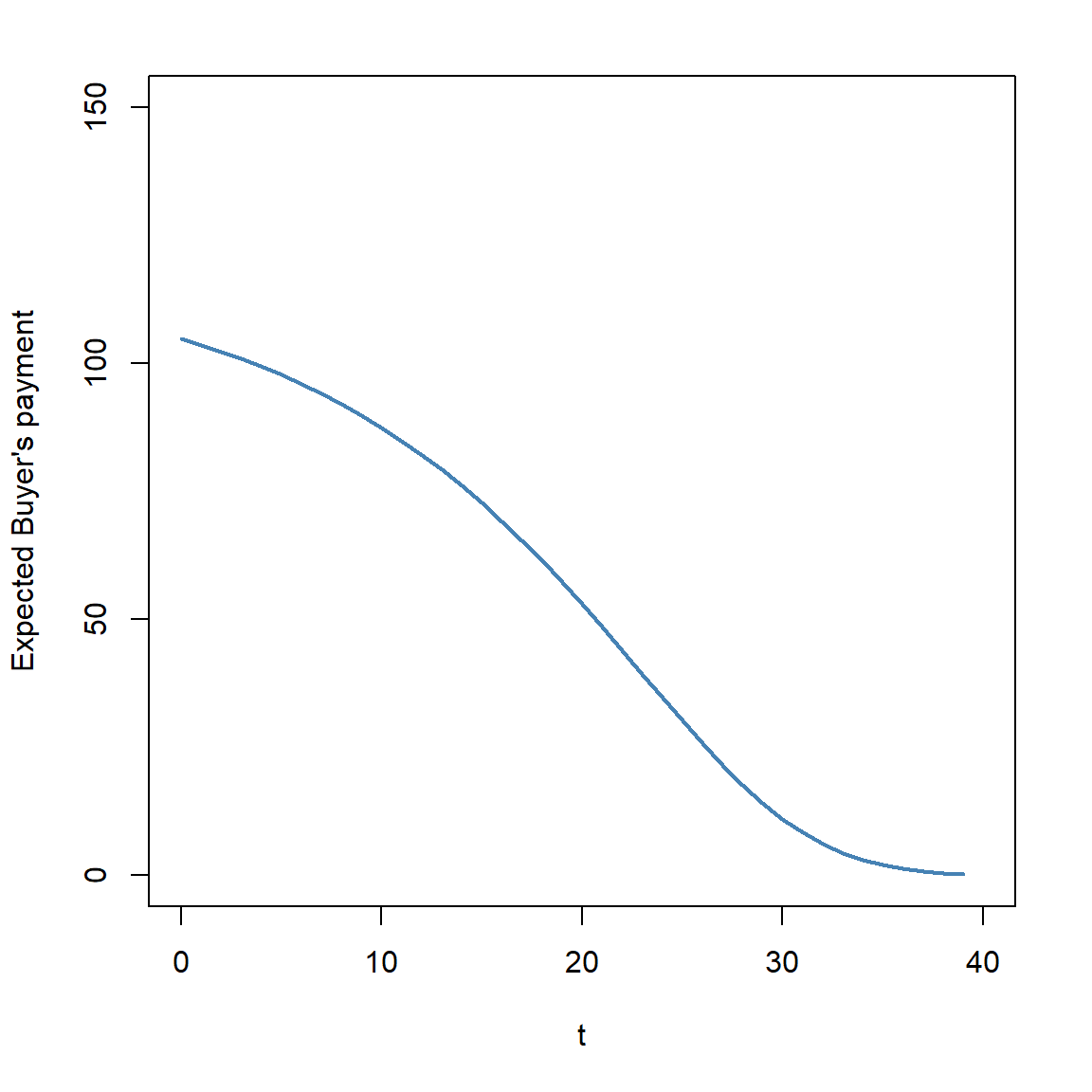} }%
    \qquad
    \subfloat[Dynamic contract]{\includegraphics[width=7.5cm]{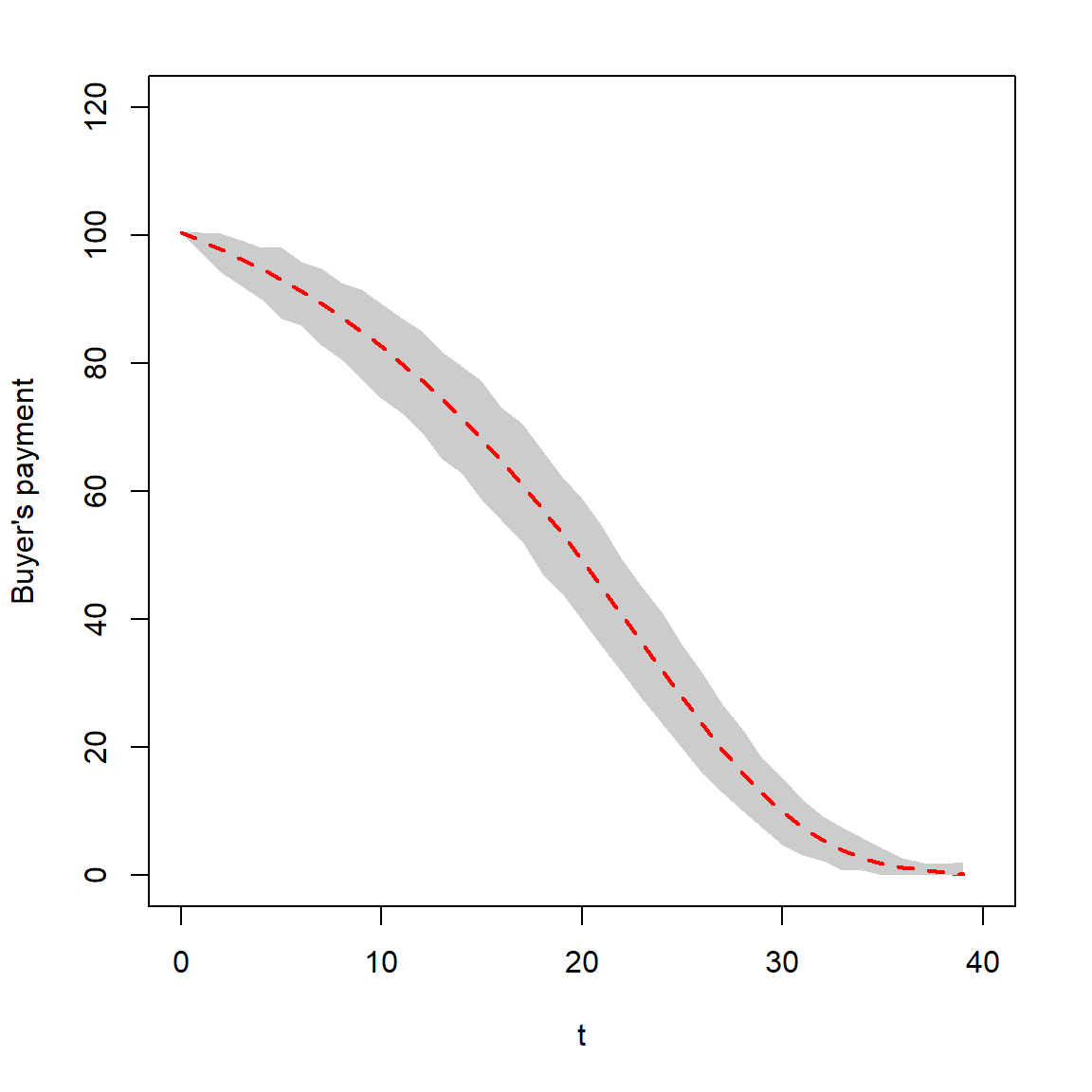} }
    \caption{Fixed-leg payments under a full hedge, $(1+\tilde{\eta}_t)\,\hat{l}_t^{[0]}$, for the static contract (left), and 95\% confidence intervals of the full-hedge fixed-leg payments, $(1+\tilde{\eta}_t)\,\hat{l}_t^{[t-1]}$, for the dynamic contract (right).}
    \label{fig:HPPAY_t}%
\end{figure}

Finally, Figure~\ref{fig:ut_traditional} compares the buyer’s \emph{optimal} hedge ratios. For the static contract (left panel), the optimal hedge ratio is about 0.4 -- buyers do not fully hedge because of the risk loading. Under the dynamic contract (right panel), the optimal hedge ratio is initially at or near a full hedge for roughly 25 years and then declines. This pattern is driven by the evolving trade-off between the mean and variance of the buyer’s liability: early on, the variance reduction dominates, whereas later the incremental variance reduction is small relative to the expected-liability increase from continued full hedging, favoring lower hedge ratios.

\begin{figure}[htbp]%
    \centering
    \subfloat[Static contract]{\includegraphics[width=7.5cm]{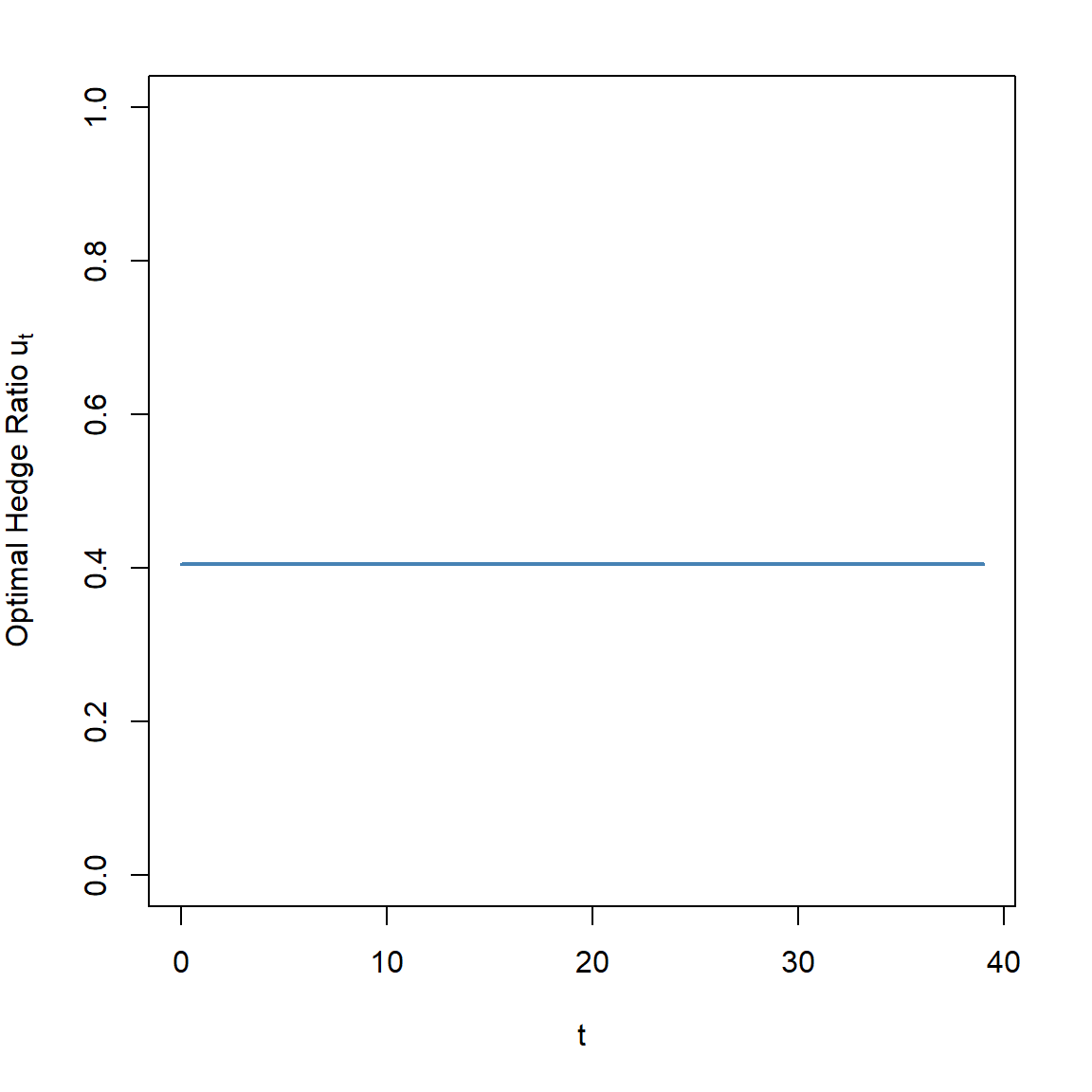} }%
    \qquad
    \subfloat[Dynamic contract]{\includegraphics[width=7.5cm]{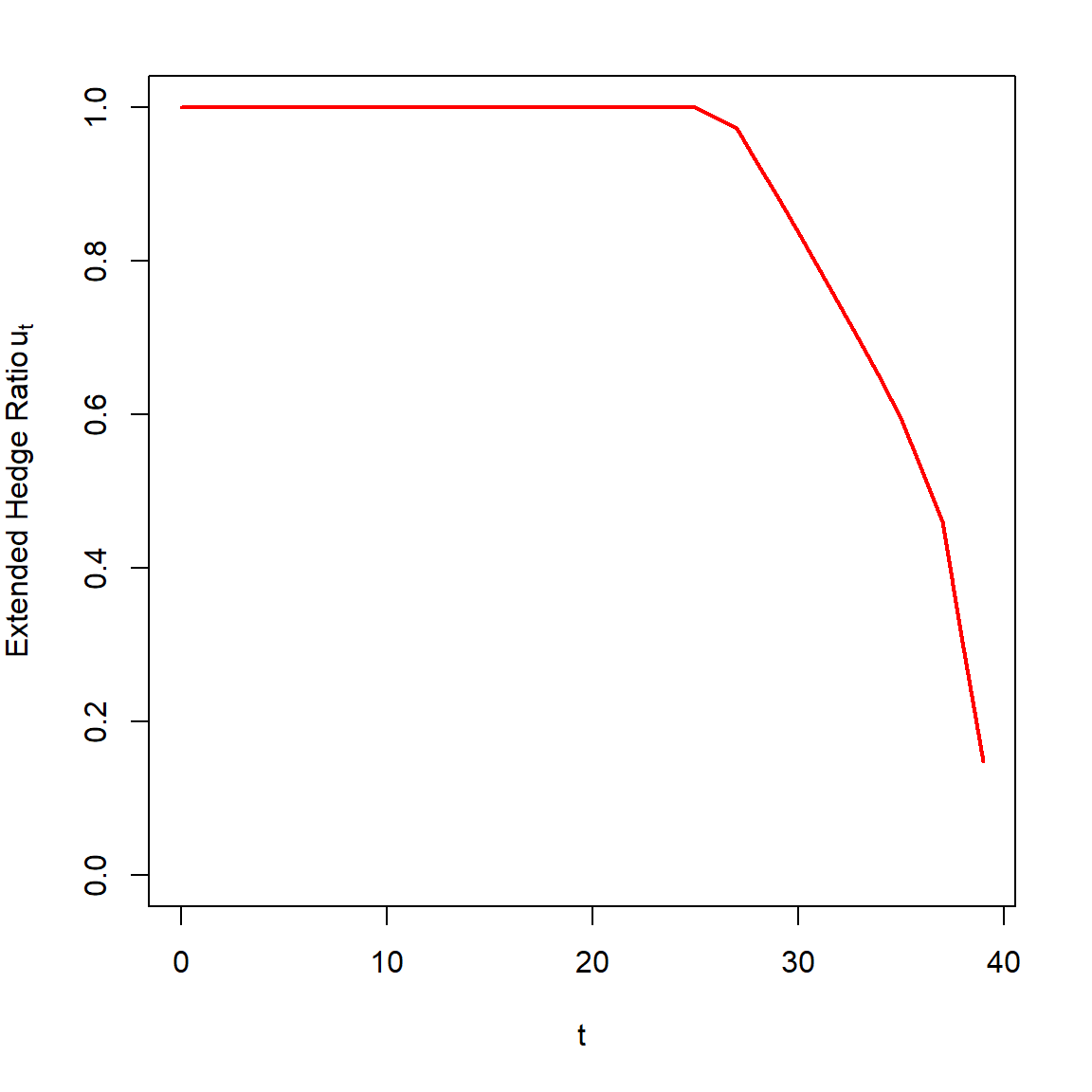} }
    \caption{Buyer’s optimal hedge ratio for static (left) and dynamic (right) contracts.}
    \label{fig:ut_traditional}%
\end{figure}

\subsection{ {Discussion and market implications}} \label{sec:discussion}

{Our numerical results point to a simple organizing principle: the preferred contract structure tracks which side is more risk-averse. When the buyer is more risk-averse than the seller, the static contract dominates, because long-horizon protection with known hedging costs aligns with the buyer’s objective, and buyers accept higher loadings for that coverage. In contrast, when the seller is more risk-averse, the dynamic contract dominates, as more risk-averse sellers avoid locking in long-dated, high-coverage exposures even at the expense of lower loadings.}

This pattern mirrors market practice: in the reinsurance channel, buyers (pension plans, annuity writers, and buy-out firms) are typically more risk-averse than global reinsurers, and long-dated, pre-committed indemnity swaps prevail; in the capital-market channel, sellers (institutional investors) are typically more risk-averse and prefer short-dated, rolled exposures, making dynamic designs the viable entry point. Thus, the model rationalizes the prevalence of long-dated indemnity swaps in reinsurance and the scarcity of long-horizon capital-market deals, alongside the emergence of short-term instruments, such as mortality CAT bond and longevity sidecars.

\section{{Seller-Side Ambiguity in the Capital-Market Setting}} \label{sec:information_asymmetry}

We now focus on the capital-market configuration highlighted in Section~\ref{sec:discussion}, with $\gamma_s>\gamma_b$. {To keep this extension tractable, we do not introduce hidden buyer types, screening, or incentive constraints.} We assume that the {reference measure} $\mathbb{P}$ serves as the most plausible proxy for the real-world measure, and thus the fixed payments of the longevity swaps, $\hat{l}_t$, are still calculated under $\mathbb{P}$. Furthermore, the fixed leg keeps the time-varying loading from~\eqref{eq:loading}, i.e.\ $(1+\tilde\eta_t)\hat l_t$ with $\tilde\eta_t=(1+\delta_L)^t\eta$.

The seller considers a collection of candidate real-world measures $\mathcal Q=\{Q_\lambda:\lambda>0\}$, with the reference measure $\mathbb P$ corresponding to $\lambda=1$. Under $Q^\lambda$, the expected $(t-s)$-year survival probability for age $x+s$ is
\begin{equation}
\mathbb E_{Q^\lambda}\!\left[\,_{t-s}p_{x+s}\right] := \left(\,_{t-s}\hat p_{x+s}\right)^{\lambda}, \label{eq:Qlambda_def_restate}
\end{equation}
where $\,_{t-s}\hat p_{x+s}$ is the estimated survival probability under the reference measure $\mathbb P$. 
With this construction, a $\lambda $ value \textit{greater} than 1
implies \textit{lower} survival probabilities and, consequently, a reduced 
life expectancy compared to what is predicted by the  {reference measure}. It
is important to note that this definition is merely a convenient way to
represent the expected values of  {real-world} survival probabilities under various
priors within our framework. Alternative definitions are also feasible and
can be integrated as needed.

In response to the uncertainty regarding the  {real-world} survival probabilities, we
adopt the maxmin expected utility framework \citep{gilboa1989maxmin} to find
the best strategy under the worst-case scenario. This is arguably the most
common approach in economic decision-making under ambiguity 
\citep[see,
e.g.,][]{dimmock2016ambiguity,kostopoulos2022ambiguity}. In our context, the worst-case scenario corresponds to the prior
probability measure $\mathbb{Q}^{\lambda }$ that leads to the largest
survival probabilities of the policyholders. The optimal contracting process
in the presence of ambiguity is detailed below.

\subsection{Stackelberg Game} \label{sec:ambiguity}

For the \textit{static contract}, given a prior probability measure $\mathbb{%
Q}^{\lambda }$, the buyer selects the optimal constant hedging ratio $u$ for
a given risk loading $\eta \geq 0$ by solving: 
\begin{equation}\label{insurer objective_static ambiguity}
\sup_{u}\left\{ \mathbb{E}^{\mathbb{Q}^{\lambda }}\left[ B_{T}(u,\eta )%
\right] -\frac{\gamma _{b}}{2}\mathrm{Var}^{\mathbb{Q}^{\lambda }}\left[
B_{T}(u,\eta )\right] \right\} .
\end{equation}%
The seller, considering the buyer's optimal response $u^{\ast }(\eta )$ and
following the maxmin principle, then calculates the optimal $\eta $ under
the worst-case scenario within $\mathcal{Q}$, the set of prior probability
measures: 
\begin{equation}
\sup_{\eta \geq 0}\inf_{\mathbb{Q}^{\lambda }\in \mathcal{Q}}\left\{ \mathbb{%
E}^{\mathbb{Q}^{\lambda }}\left[ S_{T}\bigl(u^{\ast }(\eta )\bigr)\right] -%
\frac{\gamma _{s}}{2}\mathrm{Var}^{\mathbb{Q}^{\lambda }}\left[ S_{T}\bigl(%
u^{\ast }(\eta )\bigr)\right] \right\} .  \label{MEU ambiguity}
\end{equation}

In the context of a \textit{dynamic contract}, the buyer maximizes the
mean-variance criterion with a time-varying hedge ratio: 
\begin{equation}\label{insurer objective ambiguity}
\sup_{\mathbf{u}}\left\{ \mathbb{E}^{\mathbb{Q}^{\lambda }}\left[ B_{T}(%
\mathbf{u},\eta )\right] -\frac{\gamma _{b}}{2}\mathrm{Var}^{\mathbb{Q}%
^{\lambda }}\left[ B_{T}(\mathbf{u},\eta )\right] \right\} ,
\end{equation}%
with $\mathbf{u}=(u_{t})_{t\in \{0,1,2,...,T-1\}}$. The seller's
corresponding optimization problem is then defined as: 
\begin{equation}
\sup_{\eta \geq 0}\inf_{\mathbb{Q}^{\lambda }\in \mathcal{Q}}\left\{ \mathbb{%
E}^{\mathbb{Q}^{\lambda }}\left[ S_{T}\bigl(\mathbf{u}^{\ast }(\eta )\bigr)%
\right] -\frac{\gamma _{s}}{2}\mathrm{Var}^{\mathbb{Q}^{\lambda }}\left[
S_{T}\bigl(\mathbf{u}^{\ast }(\eta )\bigr)\right] \right\} ,
\label{MEU_dynamic ambiguity}
\end{equation}%
where $\mathbf{u}^{\ast }=\bigl(u_{t}^{\ast }(\eta )\bigr)_{t\in
\{0,1,2,...,T-1\}}$.

The derivation of the optimal $\eta$ and the corresponding hedging ratios for the static contract is detailed in Appendix \ref{App： Multi AB}. A notable challenge in the dynamic contract scenario is the time inconsistency introduced by the time-varying hedging ratio, a complex aspect often encountered in stochastic mean-variance optimization problems \citep[see, e.g.,][]{bjork2010general,bjork2014mean}. To address this, we apply the game-theoretic approach proposed by \cite{bjork2010general}, formulating the optimal hedge ratio as an equilibrium strategy by solving a system of extended Hamilton-Jacobi-Bellman (EHJB) equations. Comprehensive explanations of the equilibrium strategy and associated derivations are provided in Appendix~\ref{App： Multi C}.

\subsection{Construction of $\mathcal{Q}$}  \label{sec:mQ}

To parameterize ambiguity in economically interpretable units, we map $\lambda$ to the remaining life expectancy at age $x$,
\[
\hat e_x(\lambda):=\mathbb E_{Q^\lambda}\Big[\sum_{k=1}^{T}\,_{k}p_x\Big]=\sum_{k=1}^{T} \,_{k}\hat p^{\,\lambda}_{x},\qquad \hat e_x:=\hat e_x(1),
\]
and define
\[
\mathcal Q=\Big\{Q^\lambda:\ \hat e_x(\lambda)\in\big[(1-\alpha)\hat e_x,\ (1+\alpha)\hat e_x\big]\Big\},
\]
where $\alpha\in[0,1)$ indexes the degree of ambiguity (the band is $\pm\alpha$ around $\hat e_x$ and reduces to the no-ambiguity case when $\alpha=0$). For instance, with $x=60$ and $\alpha=0.05$ (i.e., $\pm5\%$ around $\hat e_{60}$), we obtain $\lambda\in[0.81,1.21]$ under our calibration. Table~\ref{tab:lambda_alpha} reports additional ranges.

\begin{table}[htbp]
\centering
\caption{Interval for $\lambda$ at different ambiguity levels $\alpha$ (illustration with $x=60$).}
\label{tab:lambda_alpha}
\begin{tabular}{c|c}
\toprule
$\alpha$ & Interval of $\lambda$ \\
\midrule
$0.025$ & $[0.95,\,1.05]$ \\
$0.05$ & $[0.81,\,1.21]$ \\
$0.10$ & $[0.59,\,1.51]$ \\
$0.15$ & $[0.42,\,1.83]$ \\
$0.20$ & $[0.26,\,2.19]$ \\
\bottomrule
\end{tabular}
\end{table}

\subsection{Numerical Results} \label{sec:numerical_ambiguity}

We consider $(\gamma_b,\gamma_s)=(0.05,0.20)$ and three ambiguity levels: $\alpha\in\{0\%,\,2.5\%,\,5\%\}$.  Figure~\ref{fig:HP_ca} shows the objective gains for both contract types across the three levels of ambiguity. We see that objective gains from static contracts are much more sensitive to changes in ambiguity compared to dynamic contracts. Specifically, at an ambiguity level of $\alpha =5\%$, i.e., when the seller believes that the  {real-world} life expectancy of the annuitants could be 5\% larger than that indicated by the benchmark measure, static contracts no longer provide beneficial outcomes for the seller at any $\eta$, leading to a \textit{collapse of the market for static contracts}. In contrast, objective gains for dynamic contracts decline less sharply with increased ambiguity, and the market still exists at $\alpha =5\%$.

Moreover, it is intriguing to see from Figure \ref{fig:HP_ca} that the presence of ambiguity does not necessarily lead to an increase in the risk loading. For both contract types, the optimal risk loading $\eta^*$ could decrease as ambiguity rises. This \textquotedblleft counter-intuitive\textquotedblright\ result can be explained as follows. Under the maxmin expected utility valuation, the seller indeed focuses on the worst-case scenario, namely buyers with the longest remaining life expectancy (corresponding to the lowest level of $%
\lambda $). By nature, these buyers exhibit stronger demand for risk transfer. A marginal reduction in the risk loading may induce a disproportionately large increase in their demand compared to buyers with shorter remaining life expectancy. Consequently, it may be optimal for the seller to moderately lower the risk loading and offer a form of \textquotedblleft pricing discount\textquotedblright\ to encourage these buyers to choose higher hedge ratios, and ultimately enhance the seller's profit.

\begin{figure}[htbp]%
    \centering
    \subfloat[Static Contract]{\includegraphics[width=7.5cm]{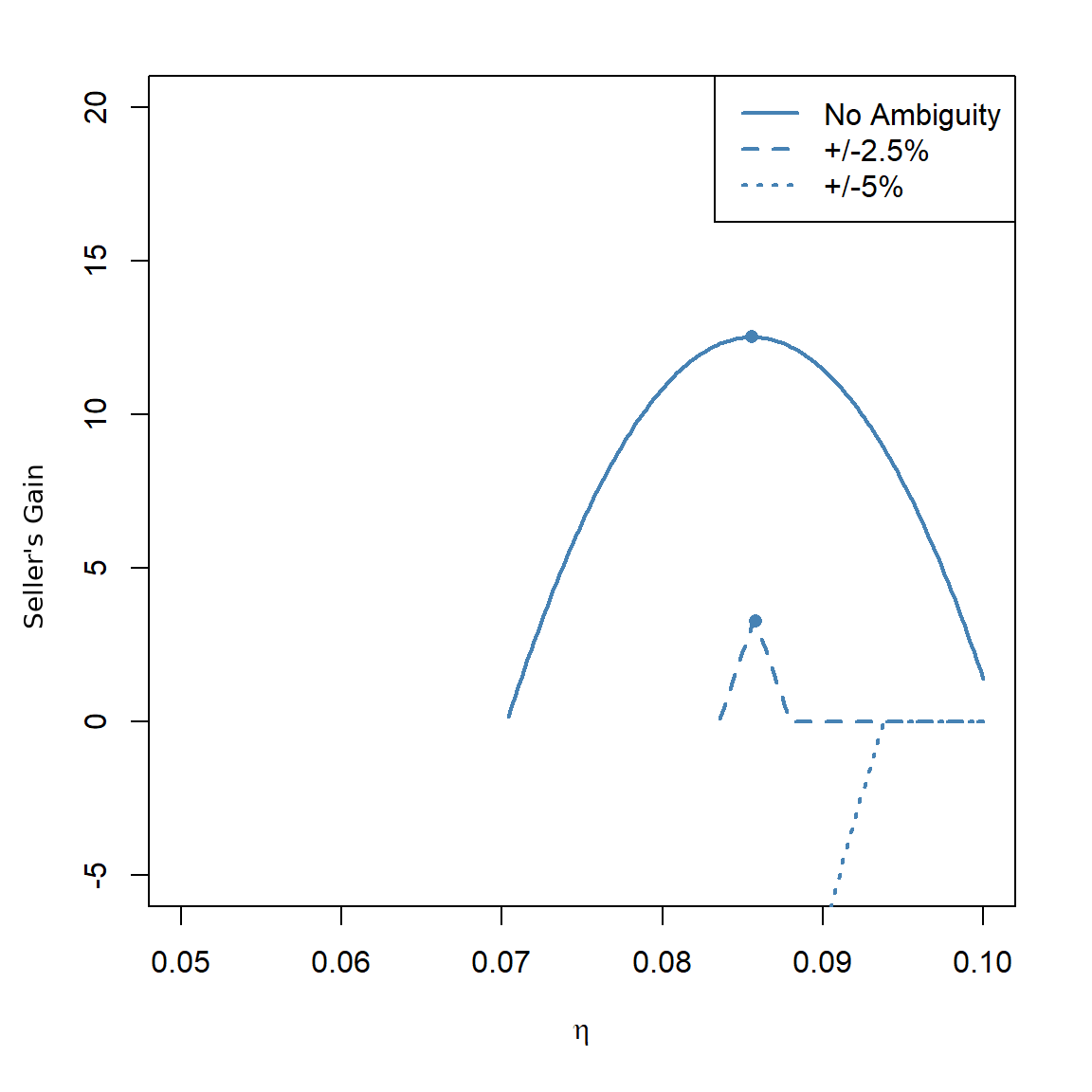} }%
    \qquad
    \subfloat[Dynamic Contract]{{\includegraphics[width=7.5cm]{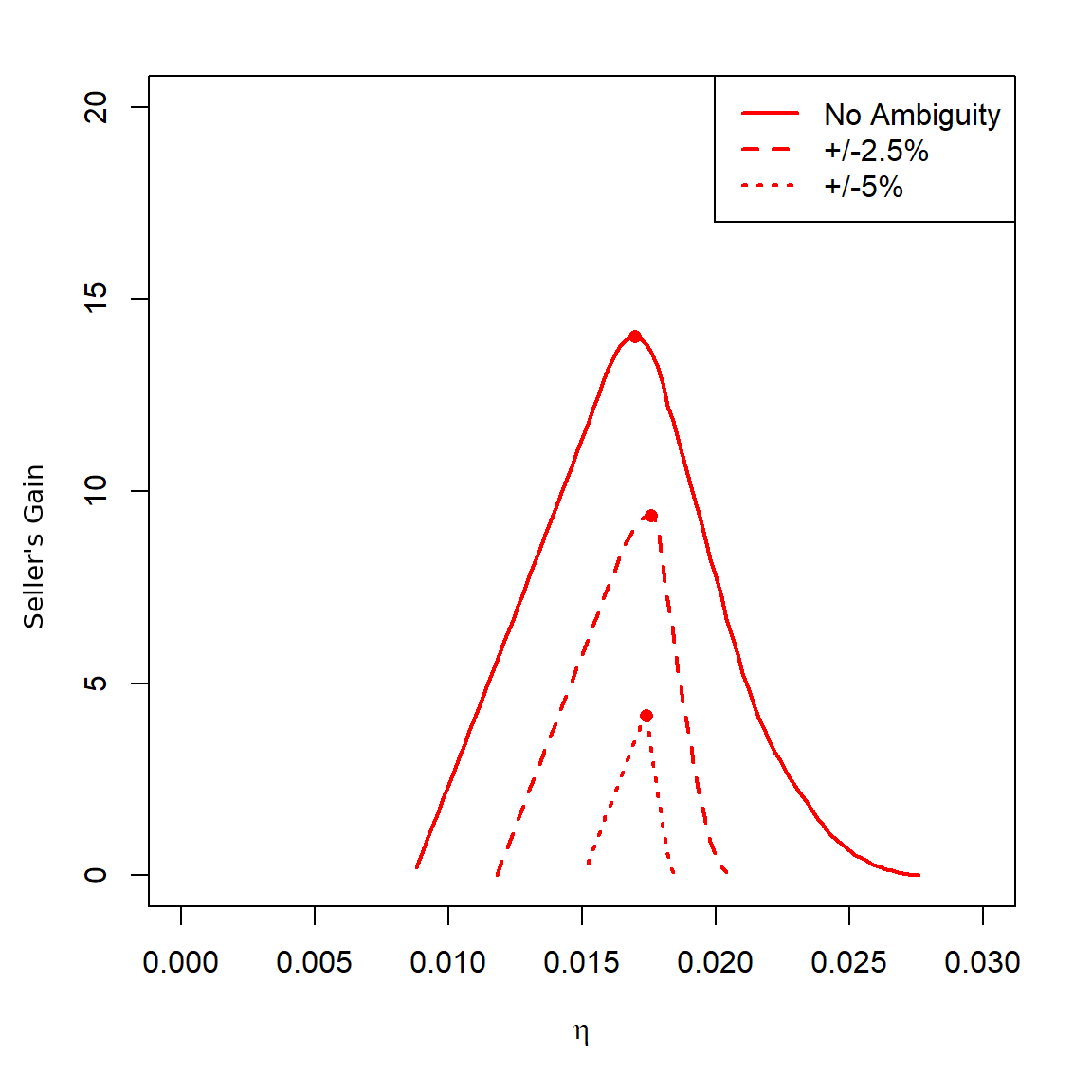} }}%
    \caption{{Objective} gain of the seller when trading the static (left) and the dynamic contract (right) across different ambiguity scenarios with $(\gamma_b,\gamma_s)=(0.05,0.20)$.}%
    \label{fig:HP_ca}%
\end{figure}

Figure \ref{fig:H_Capital_t} shows the buyer's objective gain for both static and dynamic contracts at an ambiguity level of $\alpha=2.5\%$. We see that the dynamic contract offers higher gains for the buyer across nearly all prior probability measures. This is because the more risk-averse seller requires high risk loadings to offer the static contract. Such increased loadings would sharply reduce the buyer's objective gains, making the dynamic contract a more cost-effective option despite its lower risk coverage. Specifically, for $\lambda$ values above 1.03, the buyer avoids purchasing any static contract. 

The cross-prior dispersion in Figure~\ref{fig:H_Capital_t} is also informative. Under the static contract, buyer objective gains vary substantially across priors and become negative for sufficiently large $\lambda$, so positive trade becomes concentrated in priors with higher perceived longevity exposure at the seller's chosen loading. Under the dynamic contract, buyer objective gains are more stable across priors because hedge ratios can be adjusted over time.

\begin{figure}[htbp]%
    \centering
    \includegraphics[width=7.5cm]{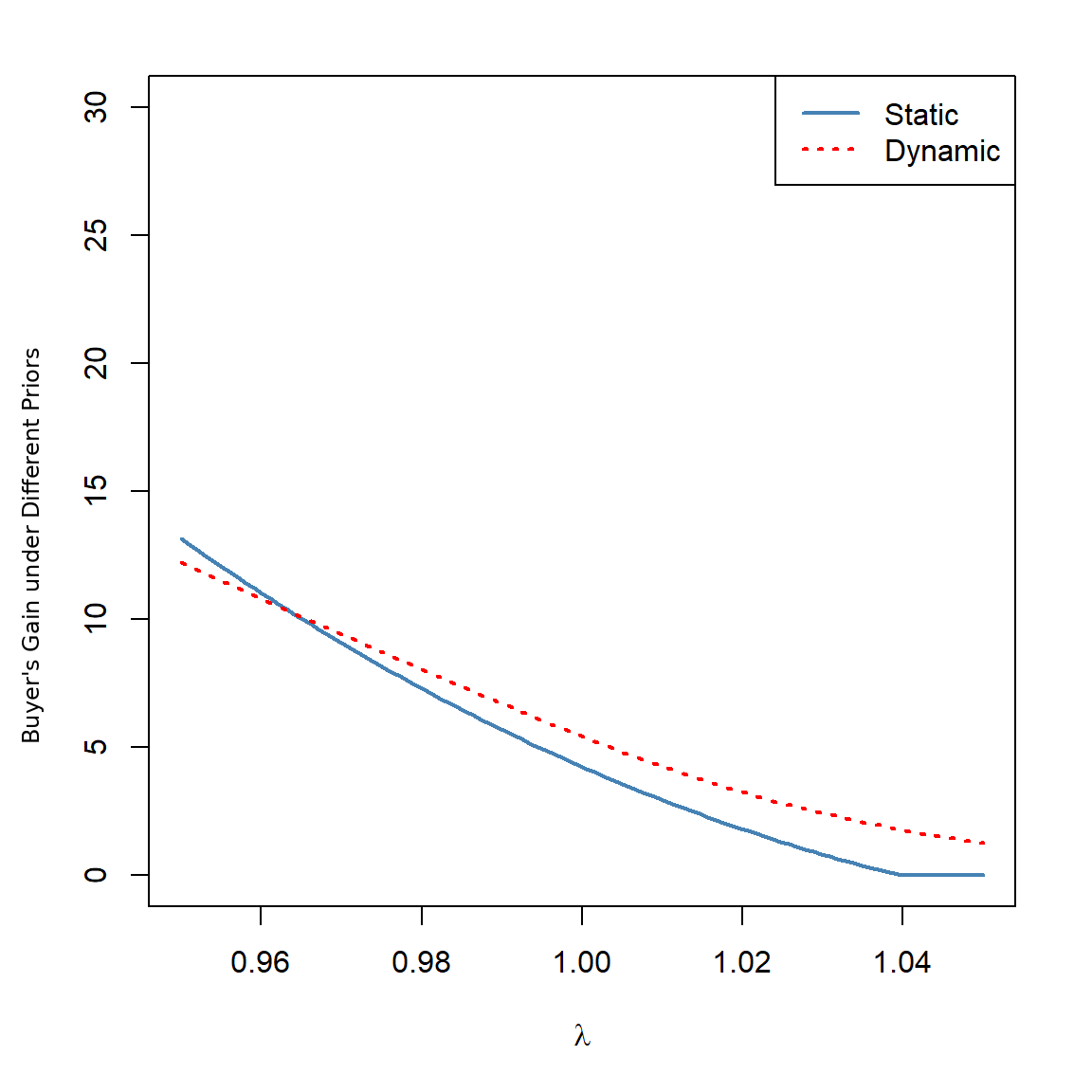}
    \caption{The buyer's gain across all prior probability measures with both contract types with the ambiguity level of $\alpha=2.5\%$ with $(\gamma_b,\gamma_s)=(0.05,0.20)$.}%
    \label{fig:H_Capital_t}%
\end{figure}

In summary, the analysis in this section shows that {the seller-side ambiguity} contributes to the scarcity of long-term longevity risk transfers in the capital market. While short-term instruments may not fully address the long-term nature of longevity risk, they are a more feasible solution than static contracts and thus offer valuable initial steps toward developing an efficient and active capital market for longevity risk transfer.

{
\section{Extension: Index-Based Longevity Swaps}\label{sec:indexbase}

This section extends the analysis of Section~\ref{sec:numerical_analysis} to the case of an index-based longevity swap, which hedges the buyer’s cohort using a broad reference population (e.g., national lives). For the discussion in this section, let $i=1$ denote the \emph{reference} (index) population and $i=2$ the buyer’s \emph{annuity} population. We work on the same filtered probability space $(\Omega,\mathcal F,\mathbb F,\mathbb P)$ as in Section~\ref{sec:model}, 
but augment the filtration $\mathbb F=\{\mathcal F_t\}$ to include mortality information from the reference population, 
so that both $l_t^{(1)}$ and $l_t^{(2)}$ are $\mathcal F_t$-measurable.

Relative to the indemnity specification in Section~\ref{sec:model}, \emph{both} legs of the swap are modified. The \emph{floating} leg settles on an index-implied survivor count, $l^{(1)}_t$, rather than the buyer’s realized survivor count $l^{(2)}_t$. The \emph{fixed} leg is priced against forecasts of the index-implied survivor count, $\hat l^{(1)}_t$, computed under the reference measure $\mathbb{P}$. These changes introduce basis risk because, in general, $l^{(1)}_t\neq l^{(2)}_t$. Due to space limitations, this section focuses on the no-ambiguity case.

The future survival counts of both populations are simulated using the age–period–cohort two-population gravity model of \citet{dowd2021hedging}. We estimate the reference population dynamics with U.S.\ unisex death rates and the annuity population with U.K.\ unisex death rates. Full details of the model and simulation are in Appendix~\ref{app:index_apc}.

\subsection{Contract specification and equilibrium} \label{sec:index_contract}

To keep units consistent with the indemnity swap, we scale the reference index by the buyer’s initial cohort size $l_0^{(2)}$. Parallel to Section~\ref{sec:contract}, the fixed leg charges $u_{t-1}(1+\tilde{\eta}_t)\,\hat l^{(1)}_t$ where, under the  {reference measure} $\mathbb{P}$,
\[
\hat{l}^{(1),[0]}_t:=\mathbb{E}^{\mathbb{P}}[l^{(1)}_t], 
\qquad
\hat{l}^{(1),[t-1]}_t:=\mathbb{E}^{\mathbb{P}}[l^{(1)}_t\mid\mathcal{F}_{t-1}].
\]
As in Section~\ref{sec:notation}, we omit superscripts $[0]$ and $[t-1]$ when no confusion arises. The index-hedged surpluses evolve as
\begin{align}
B_t^{\mathrm{ind}} &= B_{t-1}^{\mathrm{ind}}(1+r) - l_t^{(2)} 
+ u_{t-1}\bigl(l^{(1)}_t-(1+\tilde{\eta}_t)\hat{l}^{(1)}_t\bigr), \label{eq:B_index}\\
S_t^{\mathrm{ind}} &= S_{t-1}^{\mathrm{ind}}(1+r) 
- u_{t-1}\bigl(l^{(1)}_t-(1+\tilde{\eta}_t)\hat{l}^{(1)}_t\bigr). \label{eq:S_index}
\end{align}
Relative to \eqref{eq:X}–\eqref{eq:Y}, the floating leg uses $l^{(1)}_t$ (index exposure) and the fixed leg uses the index forecast $\hat{l}^{(1)}_t$ under $\mathbb{P}$ with the same time-varying loading $\tilde{\eta}_t$. Similar to the previous analysis, we adopt the parameterization $\tilde{\eta}_t=(1+\delta_L)^t\eta$.

The Stackelberg timing and mean–variance criteria are identical to Section~\ref{S:Game}, with $(B_T,S_T)$ replaced by $(B_T^{\mathrm{ind}},S_T^{\mathrm{ind}})$. 
For the static (resp.\ dynamic) swap, the buyer chooses a constant (resp.\ adapted) hedge ratio to maximize her mean-variance objective,
and the seller selects $\eta\ge 0$ anticipating the buyer’s best response.

\subsection{Numerical results}

We retain the baseline calibration used for indemnity swaps in Section~\ref{sec:numerical_analysis}, namely, $r=0.02$, $x=60$, $l_0=100{,}000$, $T=40$, and $\delta_L=0.03$. Relative to indemnity swaps, optimal risk loadings $\eta^{*}$ are systematically lower for index-based swaps under both contract structures. Intuitively, basis risk attenuates the variance-reduction from hedging, i.e., each unit of hedge transfers less idiosyncratic risk of the buyer’s book. This lowers the buyer’s demand. Anticipating this, the seller’s Stackelberg-optimal loading occurs at a lower $\eta$ than in the indemnity case. 

Figure~\ref{fig:HP_no_ambiguity_index} shows that, with a more risk-averse buyer ($(\gamma_b,\gamma_s)=(0.05,0.02)$), the seller’s objective-gain curve for the static index swap peaks at a lower $\eta$ than its indemnity counterpart in Figure~\ref{fig:HP_no_ambiguity}. The same shift is visible for dynamic designs. Figure~\ref{fig:HP_t_capital_index} shows the analogous shift when the seller is more risk-averse $(\gamma_b,\gamma_s)=(0.05,0.20)$. The bold dots in each panel mark the seller’s chosen $\eta^{*}$. Taken together, the figures show that indexation lowers $\eta^{*}$ regardless of which side is more risk-averse.

\begin{figure}[htbp!]%
    \centering
    \subfloat[Seller's {objective} gain]{\includegraphics[width=7.5cm]{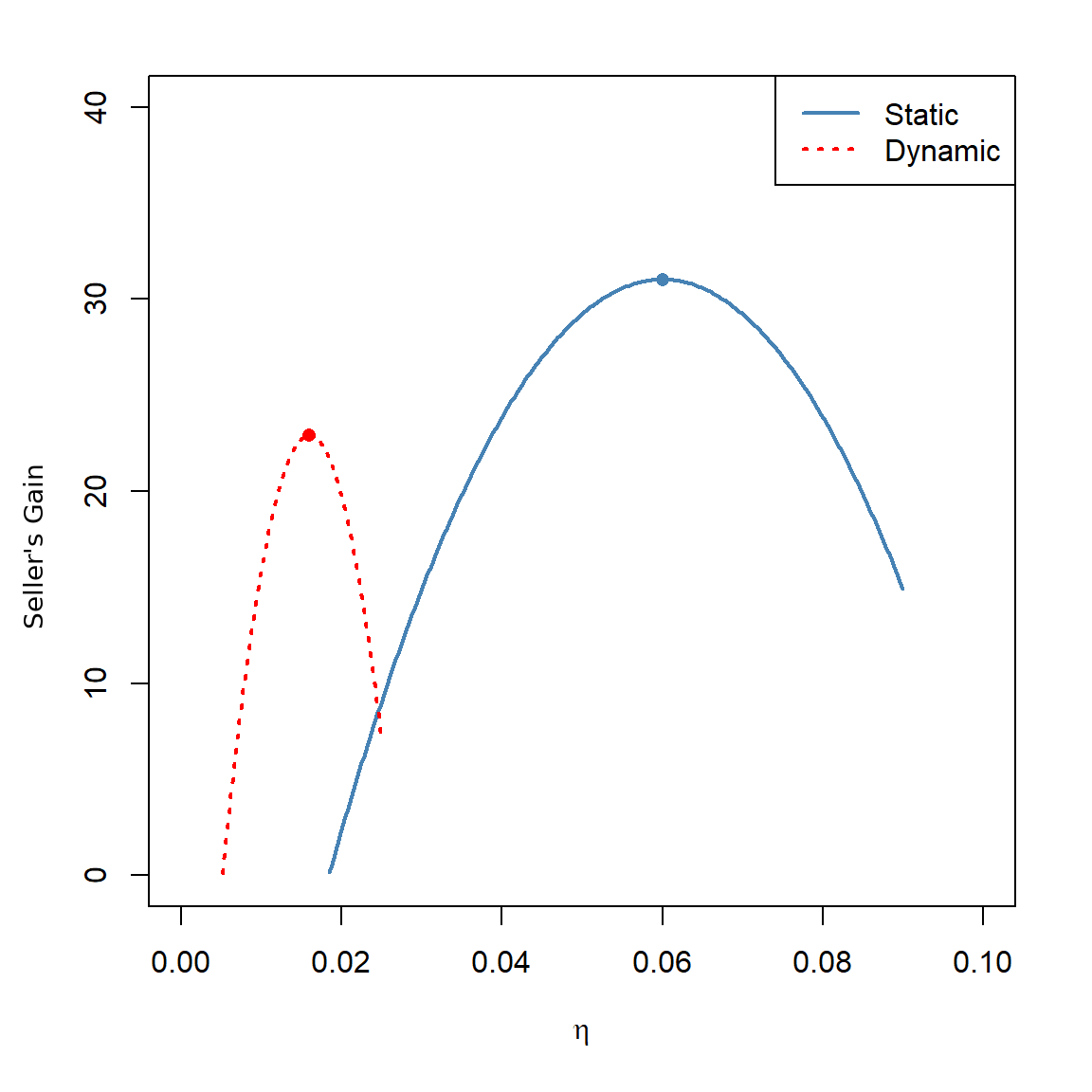} }%
    \qquad
    \subfloat[Buyer's {objective} gain]{\includegraphics[width=7.5cm]{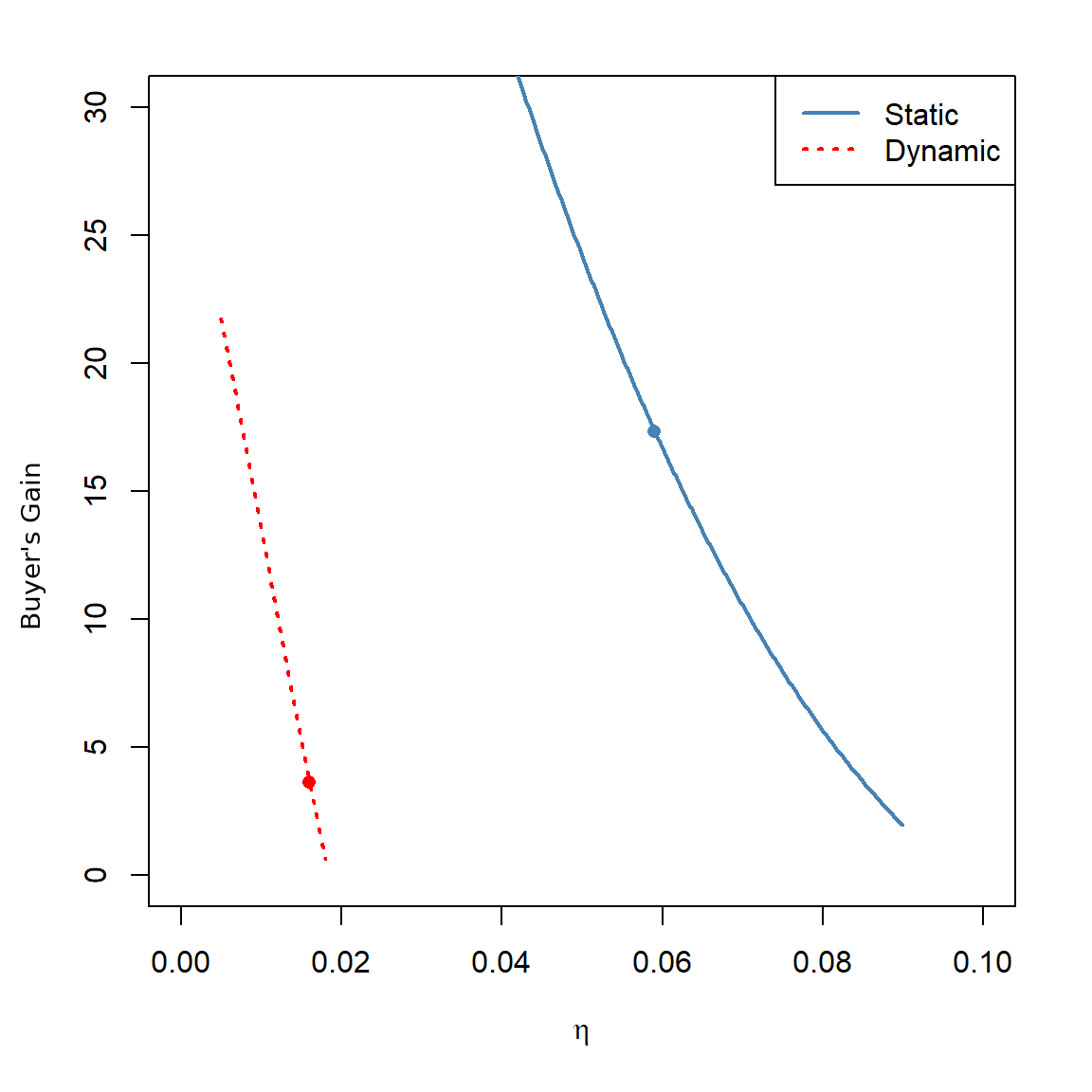} }%
    \caption{{Objective} gains of the seller (left) and the buyer (right) when trading static vs.\ dynamic \emph{index-based} contracts with $(\gamma_b,\gamma_s)=(0.05,0.02)$. The bold dot indicates the gains at $\eta^\ast$.}
    \label{fig:HP_no_ambiguity_index}%
\end{figure}

For both parties and in both risk-aversion regimes, equilibrium objective gains are lower with index-based swaps than with indemnity swaps. This is visible when comparing Figures~\ref{fig:HP_no_ambiguity_index}--\ref{fig:HP_t_capital_index} to their indemnity counterparts Figures~\ref{fig:HP_no_ambiguity}--\ref{fig:HP_t_capital}. Once basis risk is introduced, the entire objective-gain curve shifts downward and the maximal gains at $\eta^{*}$ decline for both seller and buyer. 

\begin{figure}[htbp]%
    \centering
    \subfloat[Seller's {objective} gain]{\includegraphics[width=7.5cm]{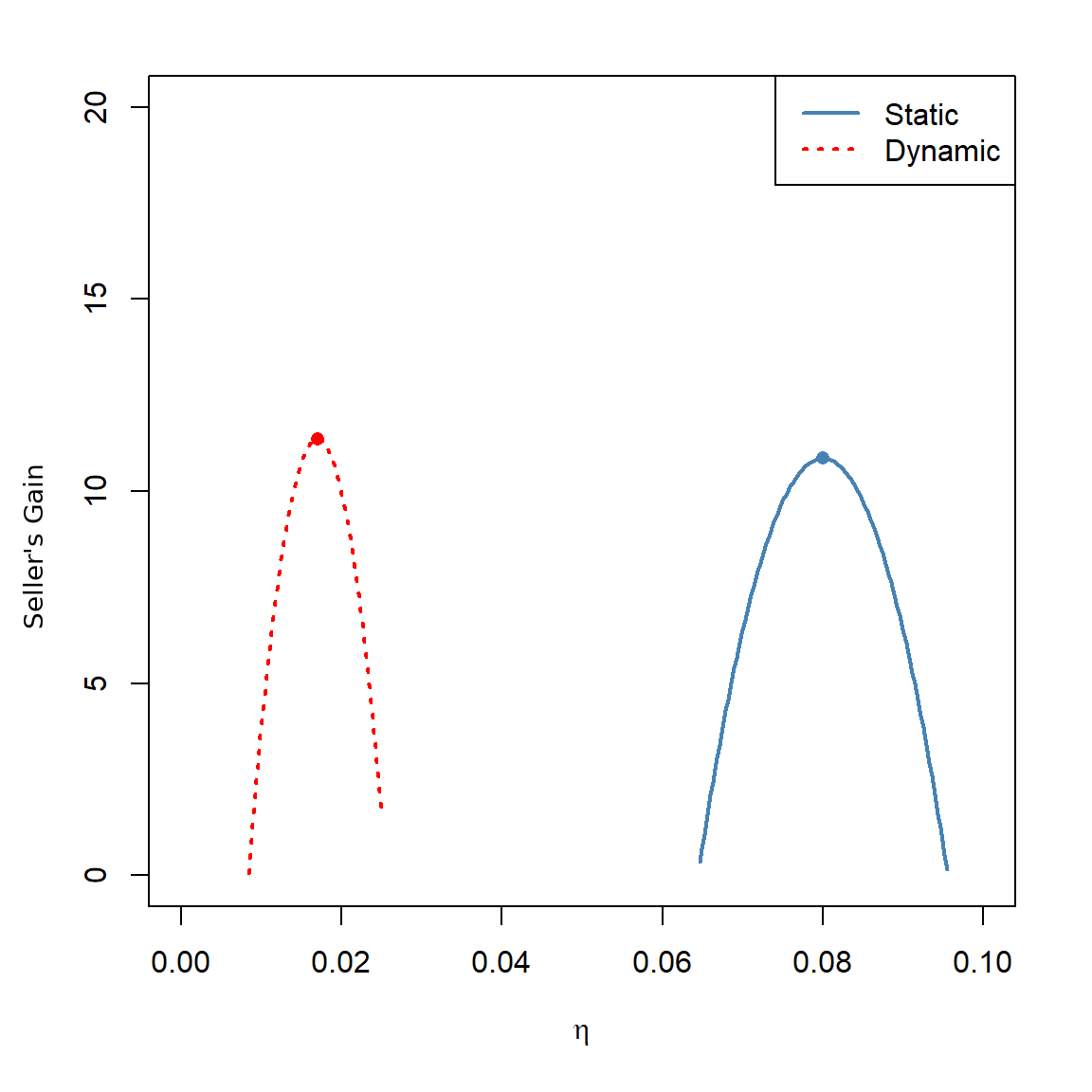} }%
    \qquad
    \subfloat[Buyer's {objective} gain]{\includegraphics[width=7.5cm]{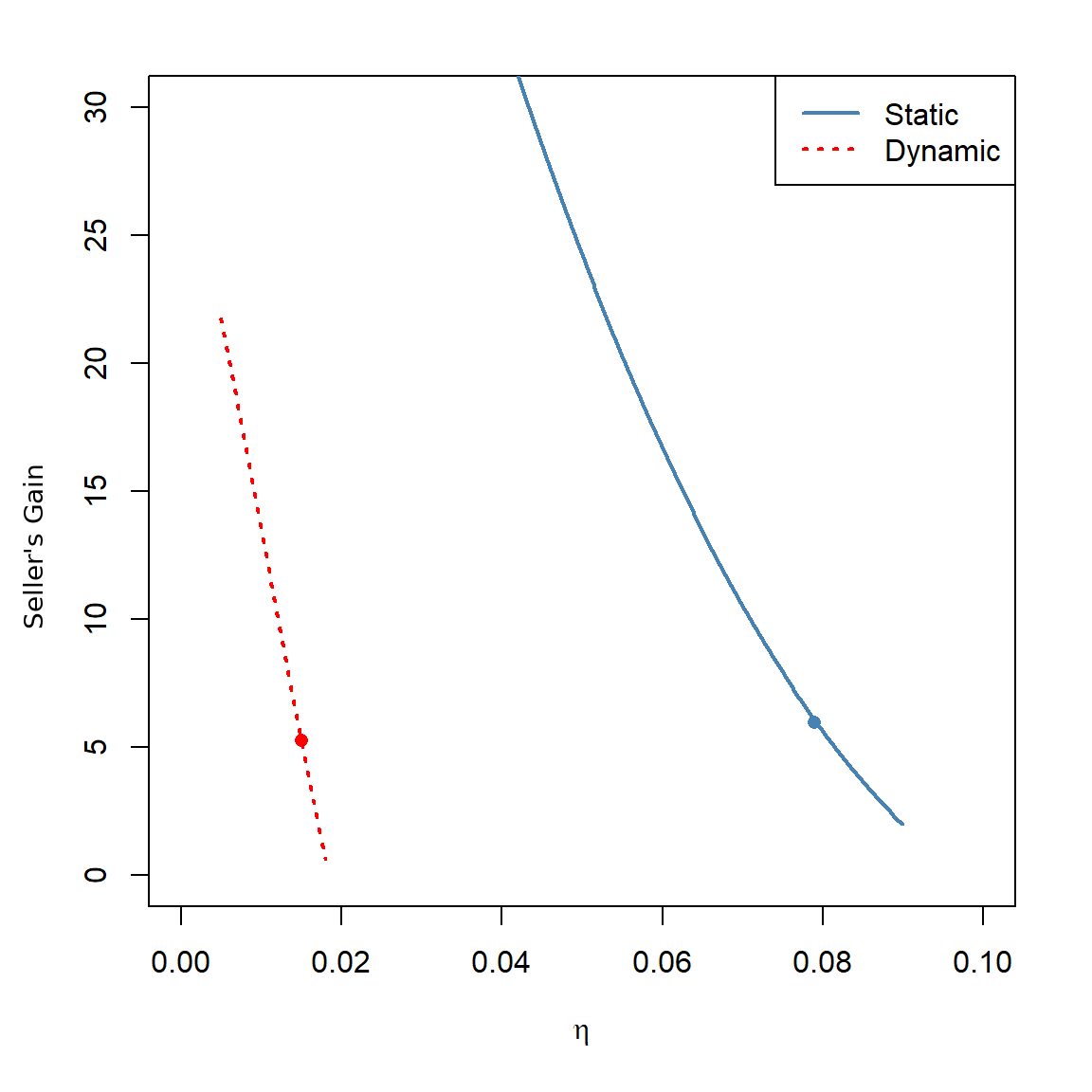} }%
    \caption{{Objective} gains of the seller (left) and the buyer (right) when trading static vs.\ dynamic \emph{index-based} contracts with $(\gamma_b,\gamma_s)=(0.05,0.2)$.}
    \label{fig:HP_t_capital_index}%
\end{figure}

Furthermore, the contract type preference from Section~\ref{sec:numerical_analysis} persists with the use of index-based swaps. When the buyer is more risk-averse than the seller, the static index-based swap dominates for the seller at equilibrium because buyers value long-horizon protection with known hedging costs and accept higher $\eta$ within the feasible range. When the seller is more risk-averse, the dynamic index-based swap dominates because sellers avoid long-dated, higher-coverage exposures and accept a lower $\eta$ in exchange for short-term rolled risk. The ordering in Table~\ref{tab:expanded_table1} is consistent with this principle. The optimal loadings and the objective gains with both the indemnity and the index-based swaps under different scenarios are summarized in Table~\ref{tab:expanded_table1}.

Two forces drive these results. First, basis risk reduces the buyer’s demand elasticity to price. For a given $\eta$, a unit increase in $u$ delivers a smaller variance reduction than under indemnity. This shifts the buyer’s best response downward and lowers the seller’s optimal $\eta^{*}$. Second, commitment and timing work as in the indemnity case. Static designs pre-commit both the hedge ratio and the fixed-leg forecasts at time 0, which provides predictable coverage that more risk-averse buyers value. Dynamic designs reset the hedge ratio and the fixed leg one step ahead, which limits long-horizon exposure and suits more risk-averse sellers. The joint effect is visible in Figure~\ref{fig:ut_traditional_index}. The dynamic $u^{*}_t$ remains close to a full hedge early when annual variance reduction dominates, then tapers as the horizon shortens and the incremental benefit falls relative to the expected cost at positive $\eta$.

\begin{figure}[htbp]%
    \centering
    \subfloat[Static contract]{\includegraphics[width=7.5cm]{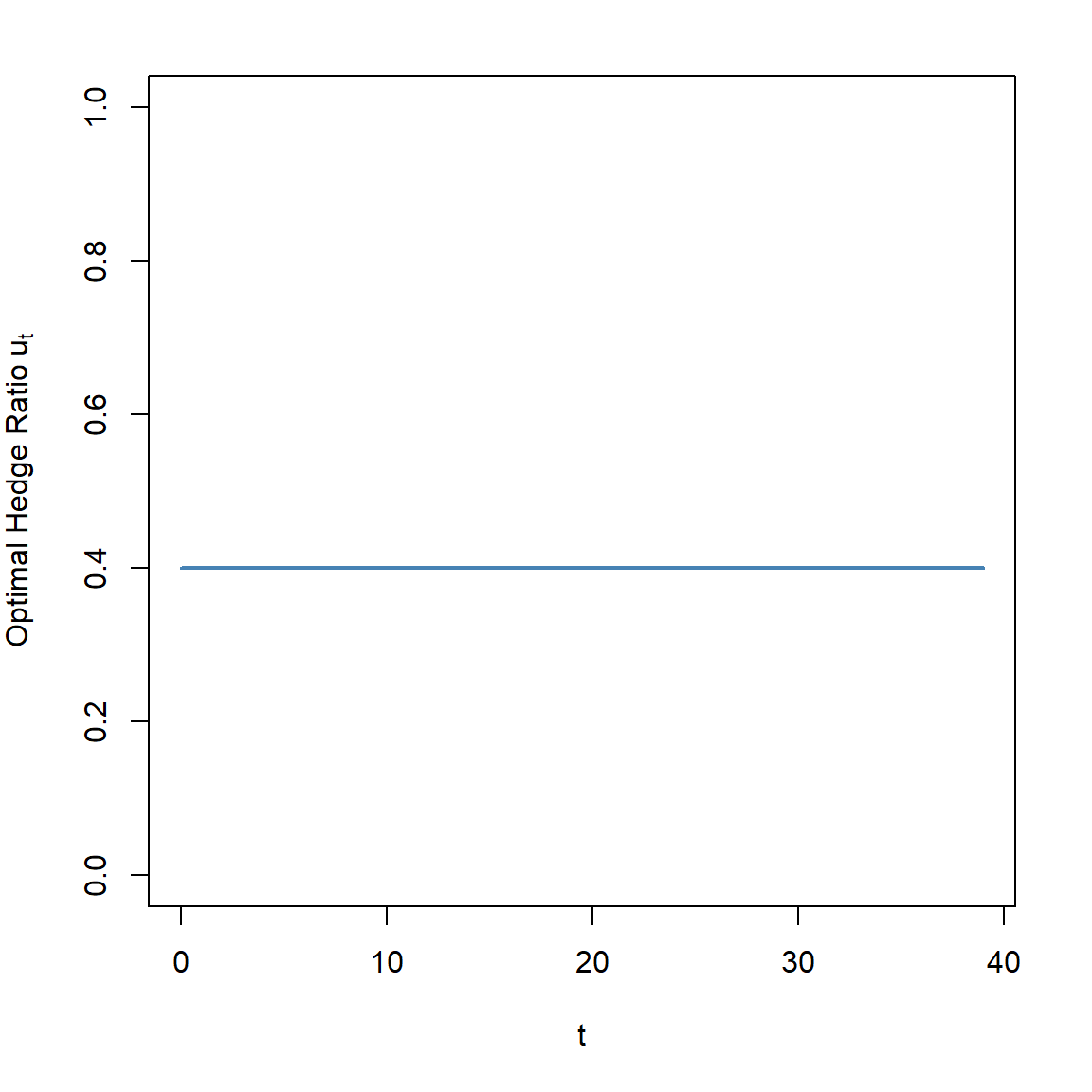} }%
    \qquad
    \subfloat[Dynamic contract]{\includegraphics[width=7.5cm]{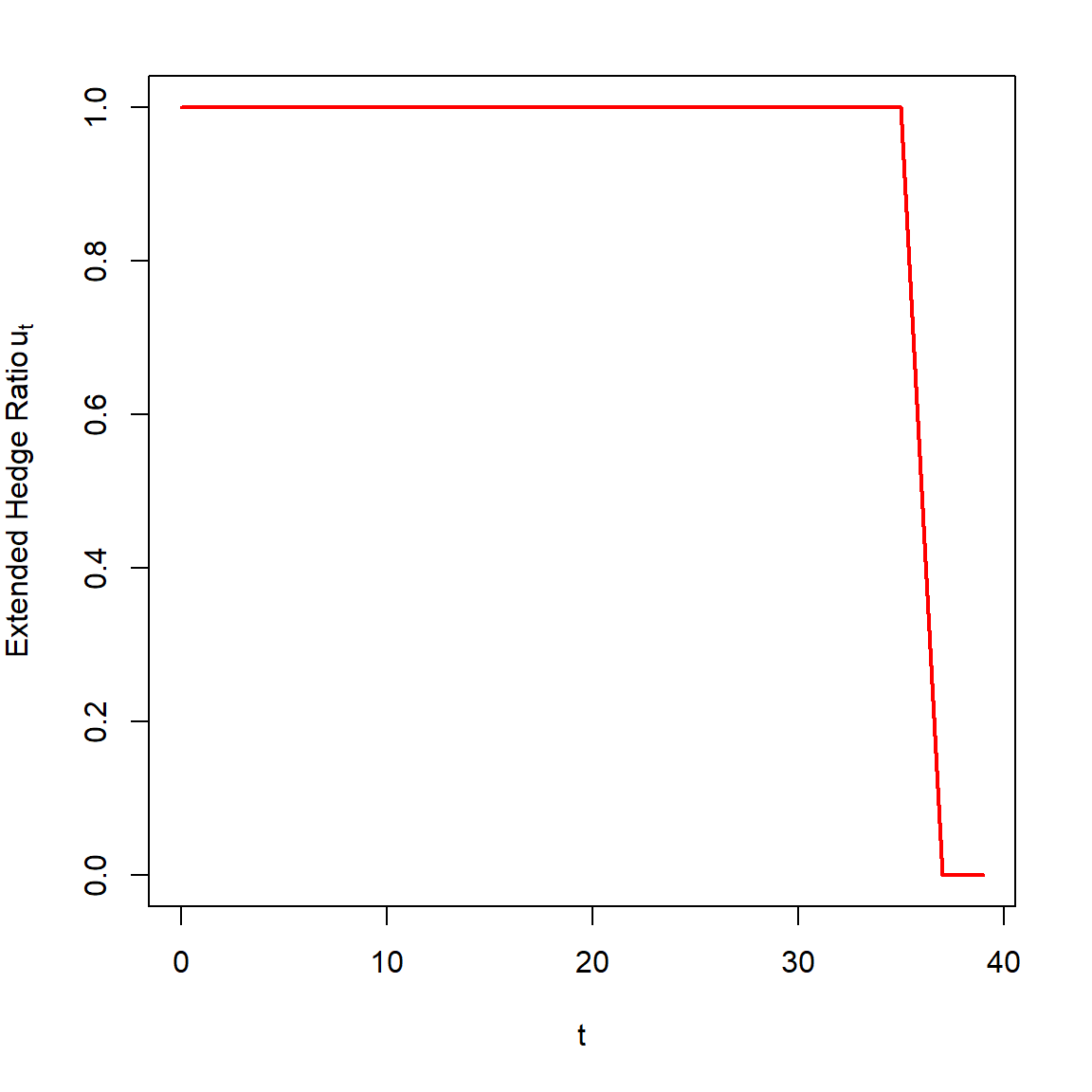} }
    \caption{Buyer’s optimal hedge ratio for static (left) and dynamic (right) \emph{index-based} contracts.}
    \label{fig:ut_traditional_index}%
\end{figure}

In practice, index-based structures can broaden participation by lowering information requirements and by standardizing exposure. The efficiency penalty from basis risk compresses surplus for both sides. To summarize, our analysis predicts lower loadings than indemnity for otherwise identical trades, similar qualitative hedge paths, and an unchanged ordering of static versus dynamic preferences across risk-aversion regimes.

\begin{table}[t]
\centering
\caption{Optimal loading $\eta^*$ and {objective} gains under indemnity vs.\ index-based swaps.}
\label{tab:expanded_table1}
\begin{tabular}{llccccccc}
\toprule
Risk aversion & Swap type
& \multicolumn{3}{c}{\textbf{Indemnity}} 
& \multicolumn{3}{c}{\textbf{Index-based}} \\
\cmidrule(lr){3-5} \cmidrule(lr){6-8}
& & $\eta^*$ & Seller gain & Buyer gain 
& $\eta^*$ & Seller gain & Buyer gain \\
\midrule
$(0.05,\,0.02)$ & Static  & 0.0639 & 33.56 & 21.56 & 0.0608 & 31.02 & 17.33 \\
$(0.05,\,0.02)$ & Dynamic & 0.0215 & 24.75 &  4.69 & 0.0167 & 22.91 &  3.62 \\
\addlinespace
$(0.05,\,0.20)$ & Static  & 0.0912 & 12.50 &  4.31 & 0.0790 &  10.86 &  5.95 \\
$(0.05,\,0.20)$ & Dynamic & 0.0234 & 14.03 &  5.28 & 0.0170 & 11.35 &  5.26 \\
\bottomrule
\end{tabular}
\end{table}
}

\section{Conclusion} \label{sec:conclusion}

This study develops an economic framework for longevity risk transfer between institutional buyers and sellers with different risk aversions. Comparing static (long‑dated, pre‑committed) and dynamic (short‑dated, rolled) swaps in a {Stackelberg sequential pricing framework}, we show that contract preferences follow differences in risk aversion: static contracts are preferred when the buyer is more risk averse, while dynamic contracts are preferred when the seller is more risk averse. These patterns help explain why long‑dated longevity swaps dominate in reinsurance, while short‑dated structures are more prevalent in capital markets.

In the capital‑market configuration, modeling seller‑side ambiguity confirms the robustness of dynamic designs and highlights the fragility of static trades: even moderate ambiguity can eliminate the static market while leaving dynamic contracts viable. Thus, while short‑term instruments do not fully match long‑horizon liabilities, they remain a practical route to expanding market capacity.

We also extend the analysis to index‑based swaps. Relative to indemnity swaps, optimal loadings are lower and objective gains are smaller for both parties, while the static–dynamic preference pattern is unchanged. Hence, while the use of index-based swaps can potentially broaden participation, is unlikely to overturn the contract-type preference shaped by risk aversion.

{Future work could relax the studied Stackelberg sequential pricing framework by allowing for bargaining, competitive intermediation, or endogenous contract menus.} It would also be interesting to compare additional instruments, such as index-based annuity forwards or option-type survivor derivatives, within the same framework, and to replace the variance penalty with a VaR/SCR objective more closely aligned with Solvency II practice.

\vspace{1cm}
\noindent \textbf{Acknowledgement} The authors acknowledge the support of the Natural Sciences and Engineering Research Council of Canada (NSERC), [RGPIN-04011] and [RGPIN-04338].

\bibliographystyle{elsarticle-harv}
\bibliography{reference}

\newpage 
\appendix
	
\noindent {\huge \textbf{Appendix}	}
	
\numberwithin{equation}{section}	
\numberwithin{figure}{section}

\section{Risk-Neutral Seller} \label{sec:risk_neutral_seller}

This section examines the robustness of our findings regarding risk aversion parameters by considering the case where the seller is risk-neutral ($\gamma_s=0$), while buyer's risk aversion $\gamma_b$ remains at 0.05. This setup indicates a larger difference in risk aversion between the two parties. 

Figure~\ref{fig:HP_no_ambiguity_risk_neutral} displays the objective gains for both parties under this configuration. Similar to observations from Figure~\ref{fig:HP_no_ambiguity}, trading the static contract yields significantly greater objective gains for both parties compared to the dynamic contract. Notably, when the seller is risk-neutral, these gains are even more pronounced than in scenarios where the seller is risk-averse.

\begin{figure}[htbp!]%
    \centering
    \subfloat[Seller's {objective} gain]{\includegraphics[width=7.5cm]{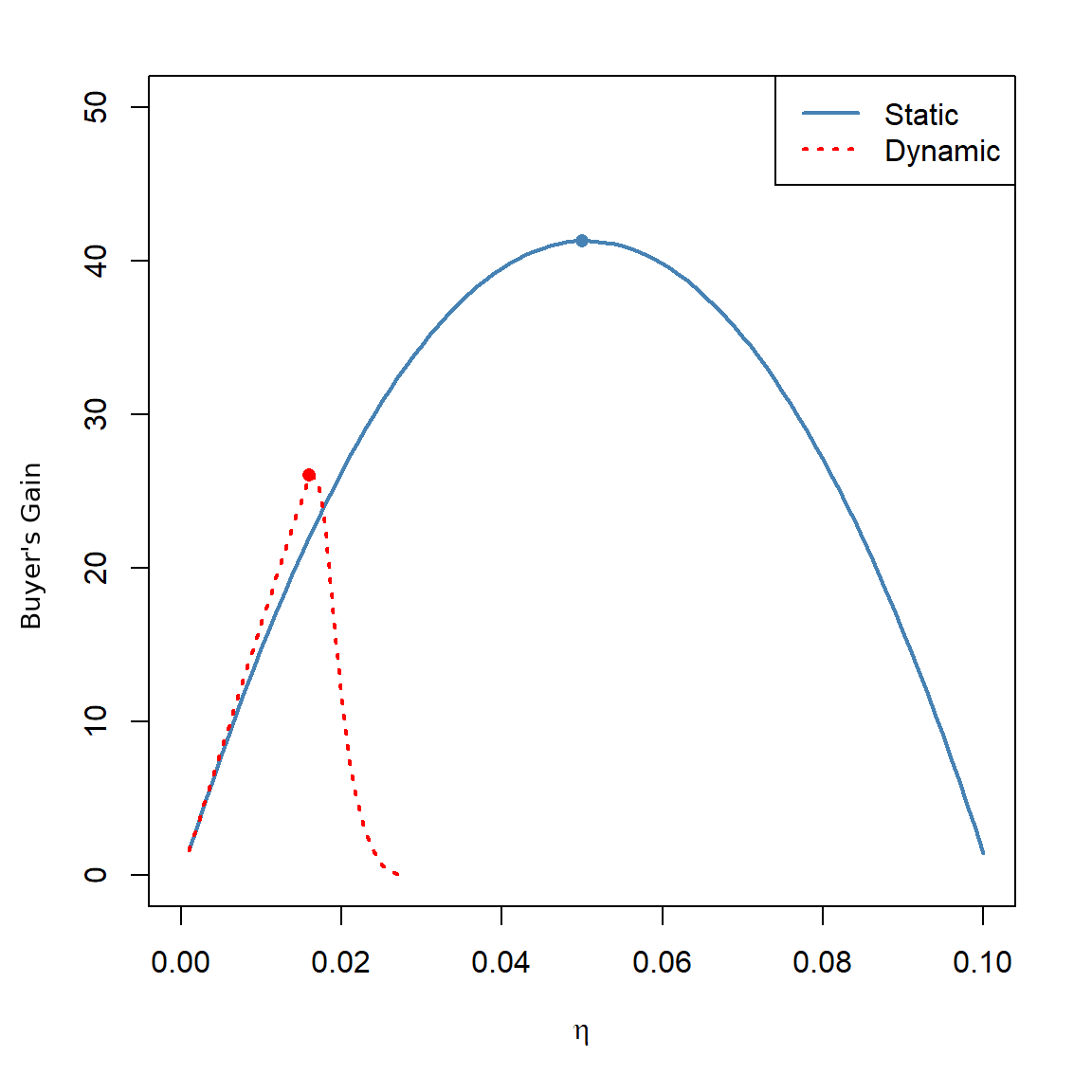} }%
    \qquad
    \subfloat[Buyer's {objective} gain]{{\includegraphics[width=7.5cm]{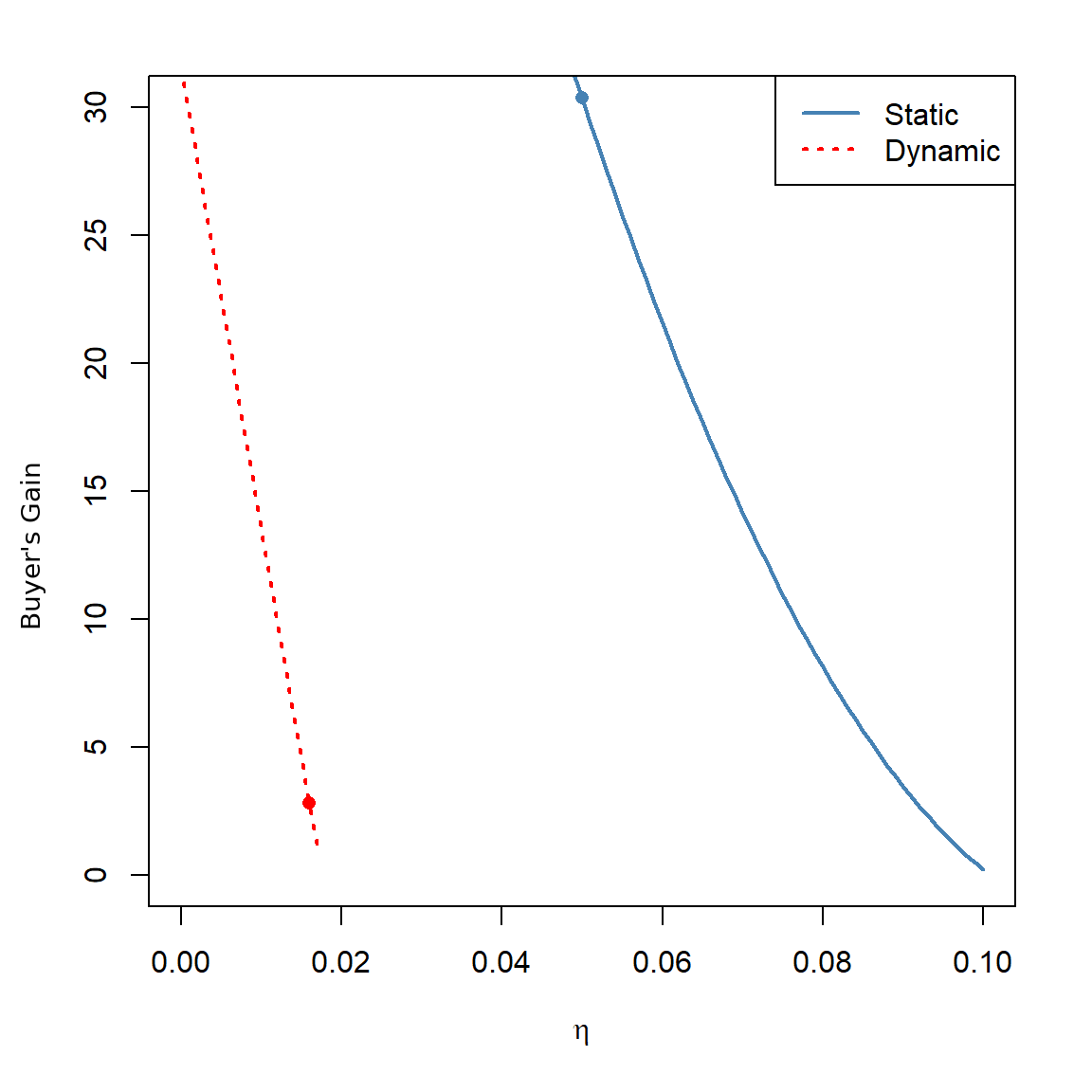} }}%
    \caption{{Objective} gain of the seller (left panel) and the buyer (right panel) when trading the static and the dynamic contract when seller is risk-neutral.}%
    \label{fig:HP_no_ambiguity_risk_neutral}%
\end{figure}

\section{Optimal Contracts in the Presence of Ambiguity}

This section discusses the derivation of optimal contract parameters in the presence of ambiguity for both the static (Section~\ref{App： Multi AB}) and the dynamic (Section~\ref{App： Multi C}) contract.

\subsection{Static Contract}\label{App： Multi AB}
\begin{theorem}\label{constant_u}
Suppose we are given the seller’s premium loading $\eta$, the buyer's optimal hedge ratio with the \textbf{static contract} equals:
\begin{align}\label{optimal u : A}
    u^{\ast}(\cdot;\eta,\lambda) = 0\vee\left(1 - \frac{C^{\lambda}_1}{\gamma_b D^{\lambda}_1}\right)\wedge 1,
\end{align}
where 
\begin{align}
    C^{\lambda}_1 &:= l_{0}\sum_{i=1}^{T} (1+r)^{T-i} \left((1+\eta){}_{i}\hat{\mathrm{p}}_{x}-{}_{i}\hat{\mathrm{p}}_{x}^{\lambda}\right),\\
    D^{\lambda}_1 &:=    \sum_{i=1}^{T}\sum_{j=1}^{T}(1+r)^{2T-(i+j)}\mathrm{Cov}^{\lambda}(l_{i},l_{j}).
\end{align}
 {The symbols $a\wedge b$ and $a\vee b$ denote $\min\{a,b\}$ and $\max\{a,b\}$, respectively.}
\end{theorem}
\begin{proof}
    We consider a constant hedge ratio over the term of the contract. For the \textbf{static contract}, the fixed payments at time $t$, $\hat{l}_{t}$ is equal to ${}_{t}\hat{\mathrm{p}}_{x}l_{0}$ and the buyer's surplus with a constant hedging strategy $u_t = u \in [0,1]$. In this case, the wealth process $\{B_{t}(u)\}_{t\in \{1,2,...,T\}}$ in Equation \eqref{eq:X} leads to the following terminal wealth $B_T(u)$:
\begin{align}
B_{T}(u) &= B_{0}(1+r)^{T} - \sum_{i =1}^{T}\left(l_{i} + u \bigl(  (1+\eta) {}_{i}\hat{\mathrm{p}}_{x}l_{0} - l_{i} \bigr)\right)(1+r)^{T-i}. \label{eq:T  X}
\end{align}
We then calculate the expectation and variance of the terminal wealth $B_T(u)$ respectively as:
\begin{align}
\mathbb{E}^{\mathbb{Q}^{\lambda}}[B_T(u)] &= B_0(1+r)^{T} - \sum^{T}_{i=1}\left(l_0{}_{i}\hat{\mathrm{p}}_{x}^{\lambda} + ul_0((1+\eta){}_i\hat{\mathrm{p}}_x-{}_{i}\hat{\mathrm{p}}_{x}^{\lambda})\right)(1+r)^{T-i}, \label{exp: A}\\
\mathrm{Var}^{\mathbb{Q}^{\lambda}}[B_T(u)] &= (1-u)^2\sum^{T}_{i=1}\sum^{T}_{j=1}(1+r)^{2T-i-j}\mathrm{Cov}^{\lambda}(l_{i},l_{j}).\label{var: A}
\end{align}
 We then plug the expectation (\ref{exp: A}) and the variance (\ref{var: A}) into buyer's objective (\ref{insurer objective_static ambiguity}), which gives:
\begin{align*}
\sup_{u \in [0,1]}\bigg \{&  B_0(1+r)^{T} - \sum^{T}_{i=1}l_0{}_{i}\hat{\mathrm{p}}_{x}^{\lambda}(1+r)^{T-i} -u\sum^{T}_{i=1} l_0((1+\eta){}_i\hat{\mathrm{p}}_x-{}_{i}\hat{\mathrm{p}}_{x}^{\lambda})(1+r)^{T-i}  \notag \\
& - \frac{\gamma_b}{2}(1-u)^2\sum^{T}_{i=1}\sum^{T}_{j=1}(1+r)^{2T-i-j}\mathrm{Cov}^{\lambda}(l_{i},l_{j}) \bigg \} 
\end{align*}
From the first-order condition, we find the optimal strategy $u^{\ast}$ (\ref{optimal u : A}). 
\end{proof}

\subsection{Dynamic Contract}\label{App： Multi C}

\begin{definition}
We consider a fixed control law $u$ and make the following construction
\begin{itemize}
    \item Fix an arbitrary point $(t, b)$, where $t<T$, and choose an arbitrary control value $u \in [0,1]$
    \item Now define the control law $\mathbf{u}^{u, t}$ on the time set $\{t, t+1, \ldots, T-1\}$ by setting, for any $b \in \mathbb{R}$,
$$
\mathbf{u}_k^{u, t}(b)= \begin{cases}\mathbf{u}^{\ast}_k(b) & \text { for } k=t+1, \ldots, T-1 \\ u & \text { for } k=t\end{cases}
$$
\end{itemize}
We say that $\mathbf{u}^{\ast}$ is a subgame perfect Nash equilibrium strategy if for every fixed $(t, b)$, we have
$$
\sup _{u \in [0,1]} J_t\left(b, \mathbf{u}^{u, t};\eta,\lambda\right)=J_t\left(b, \mathbf{u}^{\ast};\eta,\lambda\right),
$$
where $J_t\left(b, \mathbf{u}^{\ast};\eta,\lambda\right) = \mathbb{E}_{t}^{\mathbb{Q}^{\lambda}}\left[ B_{T}({\mathbf{u}^{\ast}})\right] -\frac{\gamma_b}{2}\mathrm{Var}_{t}^{\mathbb{Q}^{\lambda}}\left[B_{T}({\mathbf{u}^{\ast}})\right]$.
If an equilibrium control $\mathbf{u}^{\ast}$ exists, we define the equilibrium value function $V$ by
$$
V_t(b;\eta,\lambda):=J_t(b, \mathbf{u}^{\ast};\eta,\lambda).
$$ 
\end{definition}

\begin{theorem}
    We define the function sequence $\left(g_t;\eta,\lambda \right)_{t=0}^T$, where $g_t: \mathbb{R} \rightarrow \mathbb{R}$, by:
$$
g_t(b;\eta,\lambda):=\mathbb{E}^{\mathbb{Q}^{\lambda}}_{t, b}\left[B_T(\mathbf{u}^{\ast})\right] .
$$
The value function $V_t$ of the buyer's problem satisfies the recursion:
\begin{align}
V_t(b;\eta,\lambda)= & \sup _{u \in [0,1]}\left\{\mathbb{E}^{\mathbb{Q}^{\lambda}}_{t, b}\left[V_{t+1}\left(B_{t+1}(u);\eta,\lambda\right)\right]+\frac{\gamma_b}{2}\left(\mathbb{E}^{\mathbb{Q}^{\lambda}}_{t, b}\left[g_{t+1}\left(B_{t+1}(u);\eta,\lambda\right)\right]\right)^2\right. \nonumber \\
& \left.-\frac{\gamma_b}{2} \mathbb{E}^{\mathbb{Q}^{\lambda}}_{t, b}\left[g_{t+1}^2\left(B_{t+1}(u);\eta,\lambda\right)\right]\right\}, \label{ehjb} \\
V_T(b;\eta,\lambda)= & 
b \nonumber
\end{align}
where the recursion for $g$ is given by:
\begin{align}\label{recursion: g}
g_t(b;\eta,\lambda)=\mathbb{E}^{\mathbb{Q}^{\lambda}}_{t, b}\left[g_{t+1}\left(B_{t+1}(\mathbf{u}^\ast);\eta,\lambda\right)\right], \quad g_T(b;\eta,\lambda)=b
\end{align}
\end{theorem}
\begin{proof}
We rewrite the buyer's objective (\ref{insurer objective ambiguity}) as:
\begin{align}
    \mathbb{E}^{\mathbb{Q}^{\lambda}}_{t,b}\left[G\left( B_T(\mathbf{u})\right)\right]+H\left( \mathbb{E}^{\mathbb{Q}^{\lambda}}_{t,b}\left[B_T(\mathbf{u})\right]\right),
\end{align}
where $G(x)=x-\frac{\gamma_b}{2} x^2, \quad H(x)=\frac{\gamma_b}{2} x^2$ . We can now apply the Theorem 3.5 in (\cite{bjork2014theory}). After
straightforward calculations, the extended Bellman equation reduces to the equation (\ref{ehjb}). 
\end{proof}
\begin{theorem}\label{mul_u}
    For $t=0,1, \ldots, T-1$, the equilibrium reinsurance strategy $u^{*}(\cdot;\eta,\lambda)$ for (\ref{insurer objective ambiguity}) is given by:
\begin{align}
    u_t^{\ast}(\cdot;\eta,\lambda)&=
    0\vee\left(1-\frac{(1+\eta)\hat{\mathrm{p}}_{x+t}-\hat{\mathrm{p}}_{x+t}^{\lambda} }{\gamma_b(1+r)^{T-(t+1)}(1-\hat{\mathrm{p}}_{x+t}^{\lambda} )\hat{\mathrm{p}}_{x+t}^{\lambda} } -\frac{f^{\lambda}_{t+1}}{(1+r)^{T-(t+1)}} \right)\wedge 1.
   \label{equ l} 
\end{align}
The corresponding value functions are given by
\begin{equation}
\left\{\begin{array}{l}
V_t(b;\eta,\lambda)=(1+r)^{T-t} b +F^{\lambda}_t l_t \\
g_t(b;\eta,\lambda)=(1+r)^{T-t}b+f^{\lambda}_tl_t
\end{array}\right.
\end{equation}
where the coefficients $F^{\lambda}_t$ and $f^{\lambda}_t$ are determined recursively by
\begin{align*}
   F^{\lambda}_{t} &= \hat{\mathrm{p}}_{x+t}^{\lambda} f^{\lambda}_{t+1}- (1+r)^{T-(t+1)}(\hat{\mathrm{p}}_{x+t}^{\lambda} +u^{\ast}_t((1+\eta)\hat{\mathrm{p}}_{x+t} - \hat{\mathrm{p}}_{x+t}^{\lambda}  ) )\\
    &-\frac{\gamma_b}{2}((u^{*}_t-1)(1+r)^{T-(t+1)}+f^{\lambda}_{t+1})^2(1-\hat{\mathrm{p}}_{x+t}^{\lambda} )\hat{\mathrm{p}}_{x+t}^{\lambda} ,\\
   f^{\lambda}_{t} &= \hat{\mathrm{p}}_{x+t}^{\lambda} f^{\lambda}_{t+1}- (1+r)^{T-(t+1)}(\hat{\mathrm{p}}_{x+t}^{\lambda} +u^{\ast}_t((1+\eta)\hat{\mathrm{p}}_{x+t} - \hat{\mathrm{p}}_{x+t}^{\lambda}  ) ),
\end{align*}
with boundary conditions
\begin{align*}
  F^{\lambda}_{T} = 0 \quad \text{and} \quad f^{\lambda}_{T} = 0. 
\end{align*}
\end{theorem}
\begin{proof}
We assume the following ansatz,
\begin{align}
& V_t(b;\eta,\lambda)=A_t b+F^{\lambda}_{t}l_t, \label{ansatz: V} \\
& g_t(b;\eta,\lambda)=a_t b+f^{\lambda}_{t}l_t.\label{ansatz: g}
\end{align}
Using this ansatz and the
dynamics for $B_T$ (\ref{eq:X}), the extended Bellman equation (\ref{ehjb})reduces considerably, and we obtain
\begin{align*}
    A_t b+F^{\lambda}_{t} l_t&=\sup _{u \in [0,1]}\bigg\{A_{t+1} \left((1+r) b - l_t\hat{\mathrm{p}}_{x+t}^{\lambda} -ul_t\left((1+\eta)\hat{\mathrm{p}}_{x+t}-\hat{\mathrm{p}}_{x+t}^{\lambda} \right) \right)+F^{\lambda}_{t+1}\hat{\mathrm{p}}_{x+t}^{\lambda} l_t \\ 
    &-\frac{\gamma_b}{2} \left((u-1)a_{t+1}+f^{\lambda}_{t+1}\right)^2l_t(1-\hat{\mathrm{p}}_{x+t}^{\lambda} )\hat{\mathrm{p}}_{x+t}^{\lambda} \bigg\}.
\end{align*}
The above expression is concave with respect to $u$. Thus, we
obtain the optimal value of $u$ from the first-order condition. Specifically,
\begin{align*}
     u_t^{\ast}(\cdot;\eta,\lambda)&=
    0\vee\left(1-\frac{A_{t+1}\left((1+\eta)\hat{\mathrm{p}}_{x+t}-\hat{\mathrm{p}}_{x+t}^{\lambda} \right)}{\gamma_ba^2_{t+1}(1-\hat{\mathrm{p}}_{x+t}^{\lambda} )\hat{\mathrm{p}}_{x+t}^{\lambda} } -\frac{f^{\lambda}_{t+1}}{a_{t+1}} \right)\wedge 1.
\end{align*}
and inserting this into the equation above, we obtain  $A_t$ and the recursion formula for $B_t$,
\begin{align}
    A_t &= (1+r)^{T-t},\\ \nonumber
     F^{\lambda}_{t} &= -(1+r)^{T-(t+1)}\left( \hat{\mathrm{p}}_{x+t}^{\lambda} +u^{\ast}_t\left((1+\eta)\hat{\mathrm{p}}_{x+t} - \hat{\mathrm{p}}_{x+t}^{\lambda} \right)\right) +  F^{\lambda}_{t+1}\hat{\mathrm{p}}_{x+t}^{\lambda}  \\ \nonumber &-\frac{\gamma_b}{2} \left((u^{\ast}_t-1)a_{t+1}+f^{\lambda}_{t+1}\right)^2(1-\hat{\mathrm{p}}_{x+t}^{\lambda} )\hat{\mathrm{p}}_{x+t}^{\lambda} . 
\end{align}
We plug the ansatz
(\ref{ansatz: g}) and the previously derived expression for $u^{\ast}$ into the recursion (\ref{recursion: g}),
\begin{align*}
    a_t b + f^{\lambda}_{t}l_t = a_{t+1}\left((1+r)b -\hat{\mathrm{p}}_{x+t}^{\lambda} l_t-u^{\ast}_t\left((1+\eta)\hat{\mathrm{p}}_{x+t}-\hat{\mathrm{p}}_{x+t}^{\lambda} l_t\right)\right) + f^{\lambda}_{t+1}\hat{\mathrm{p}}_{x+t}^{\lambda} l_t.
\end{align*}
After identifying coefficients, this gives us the recursions,
\begin{align}
    a_t &= (1+r)^{T-t},\\ \nonumber
   f^{\lambda}_{t} &= \hat{\mathrm{p}}_{x+t}^{\lambda} f^{\lambda}_{t+1}- (1+r)^{T-(t+1)}\left(\hat{\mathrm{p}}_{x+t}^{\lambda} +u^{\ast}_t\left((1+\eta)\hat{\mathrm{p}}_{x+t} - \hat{\mathrm{p}}_{x+t}^{\lambda}  \right) \right). 
\end{align}
\end{proof}

\section{Two-Population Gravity APC Model and Simulation}\label{app:index_apc}

This section introduces the two-population APC model \citep{dowd2021hedging} used in Section~\ref{sec:indexbase}. For more detailed descriptions of the model and the estimation procedure, we refer the reader to \citep{dowd2021hedging}. We choose the U.S. uni-sex mortality data for the reference population ($i=1$), and the U.K. uni-sex mortality data for the annuity population ($i=2$). The data for both populations are sourced from the Human Mortality Database, using one-age, one-year death rates for ages 20 to 100 and years 1956 to 2020.

For two populations $i\in\{1,2\}$ with ages $x$ and calendar years $h$ (cohort $c=h-x$), let $m^{(i)}_{x,h}$ denote the central death rate. The APC observation equation is
\begin{equation}
\ln m^{(i)}_{x,h} = a_x + \beta_x\,\kappa^{(i)}_h + \gamma^{(i)}_c,\qquad c=h-x,\ i\in\{1,2\},
\end{equation}
with constraints $\sum_h\kappa^{(i)}_h=0$, $\sum_c\gamma^{(i)}_c=0$, and $\sum_x\beta_x=1$ for identifiability. Parameters are estimated by maximum likelihood under these constraints.

Let the period effects follow (discrete‑time) random walks with drift and a gravity pull from population 1 (reference) to 2 (annuity):
\begin{align}
\kappa^{(1)}_h &= \kappa^{(1)}_{h-1} + \mu^{(1)}_\kappa + \sigma^{(1)}_\kappa Z^{(\kappa,1)}_h, \notag \\
\kappa^{(2)}_h &= \kappa^{(2)}_{h-1} + \mu^{(2)}_\kappa + \phi_\kappa\big(\kappa^{(1)}_{h-1}-\kappa^{(2)}_{h-1}\big) + \sigma_{\kappa,21} Z^{(\kappa,1)}_h + \sigma_{\kappa,22} Z^{(\kappa,2)}_h, \label{eq:kappa_2}
\end{align}
with $Z^{(\kappa,1)}_h,Z^{(\kappa,2)}_h\stackrel{\text{i.i.d.}}\sim N(0,1)$, independent over $h$, and $\phi_\kappa\in[0,1)$. Similar to the single-population APCI model, the cohort effects do not need to be projected for our simulation study.

The gravity parameter $(\phi_\kappa)$ pull the smaller population’s period effect toward the larger population’s period effect, capturing long‑run co‑movement with short‑run deviations; the cross‑loadings $(\sigma_{\kappa,21},\sigma_{\kappa,22})$ allow instantaneous innovation correlation. Similar to the single-population model, future death rates, $\ln m^{(i)}_{x,h}$, will be simulated after estimating the model and simulating future $\kappa^{(1)}$'s and $\kappa^{(2)}$'s, and cohort survival probabilities will be simulated via the standard exponential link $_{t-s}p^{(i)}_{x+s}\approx \exp\{-\sum_{\mu=0}^{t-s} m^{(i)}_{x+s+\mu,\,s+\mu}\}$ for $i\in\{1,2\}$. 

\section{Mortality Model Estimation Results}

This section reports the estimated factors in both the single-population APCI model and the two-population gravity model, as well as the estimated parameters in the period-effect dynamics.

\subsection{Age-Period-Cohort-Improvement model}

Figure~\ref{fig:apci_components_single} displays the estimated factors in the single-population APCI model. The estimated $\hat{\sigma}_{\kappa}$ in Equation~\eqref{eq:kappa} is 0.02.

\begin{figure}[H]
\centering
\captionsetup[subfigure]{labelformat=parens,labelsep=space,font=small}
\begin{subfigure}{0.48\textwidth}
    \centering
    \includegraphics[width=\linewidth]{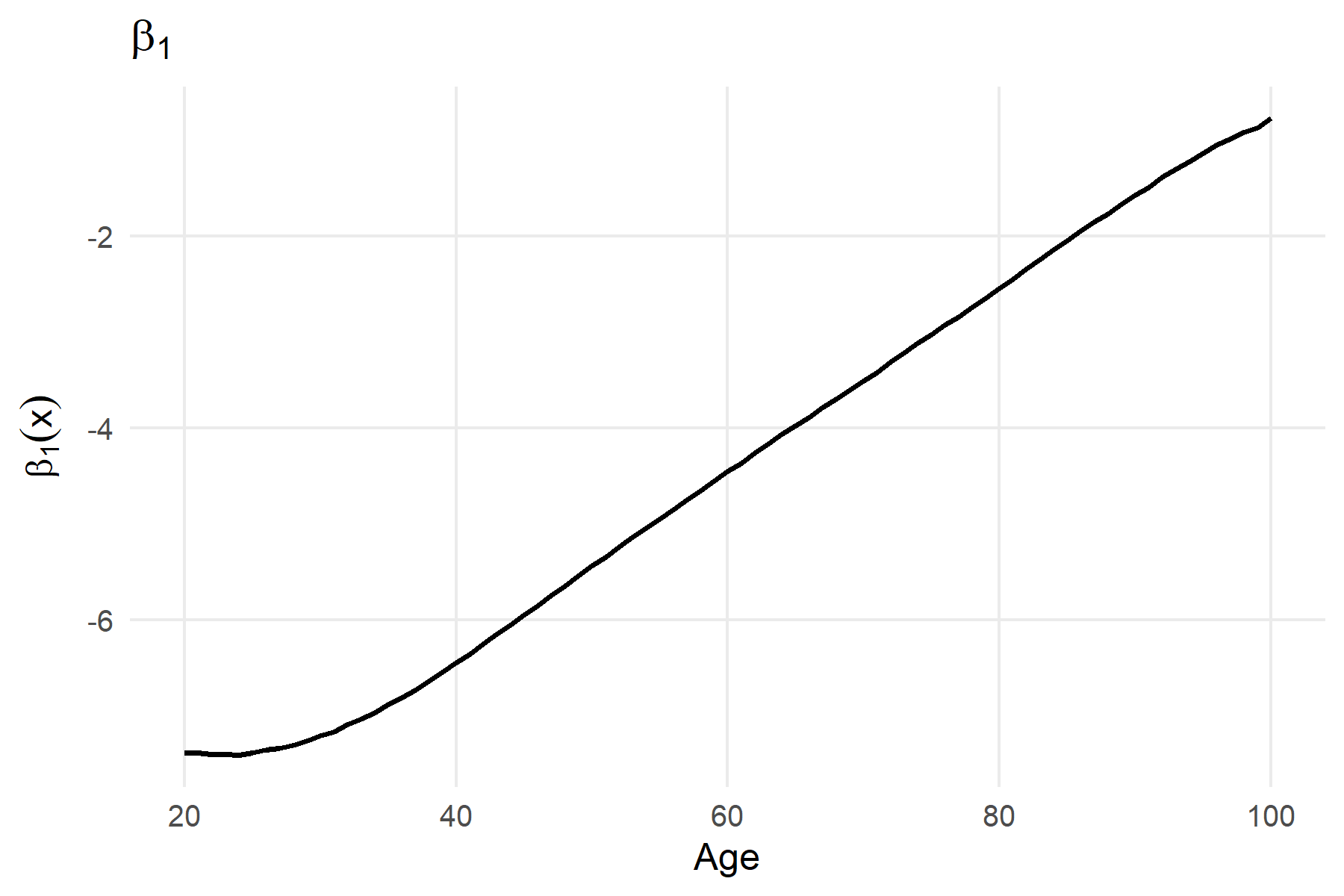}
    \caption{Age effect $\beta_{1}(x)$}
    \label{fig:beta1}
\end{subfigure}\hfill
\begin{subfigure}{0.48\textwidth}
    \centering
    \includegraphics[width=\linewidth]{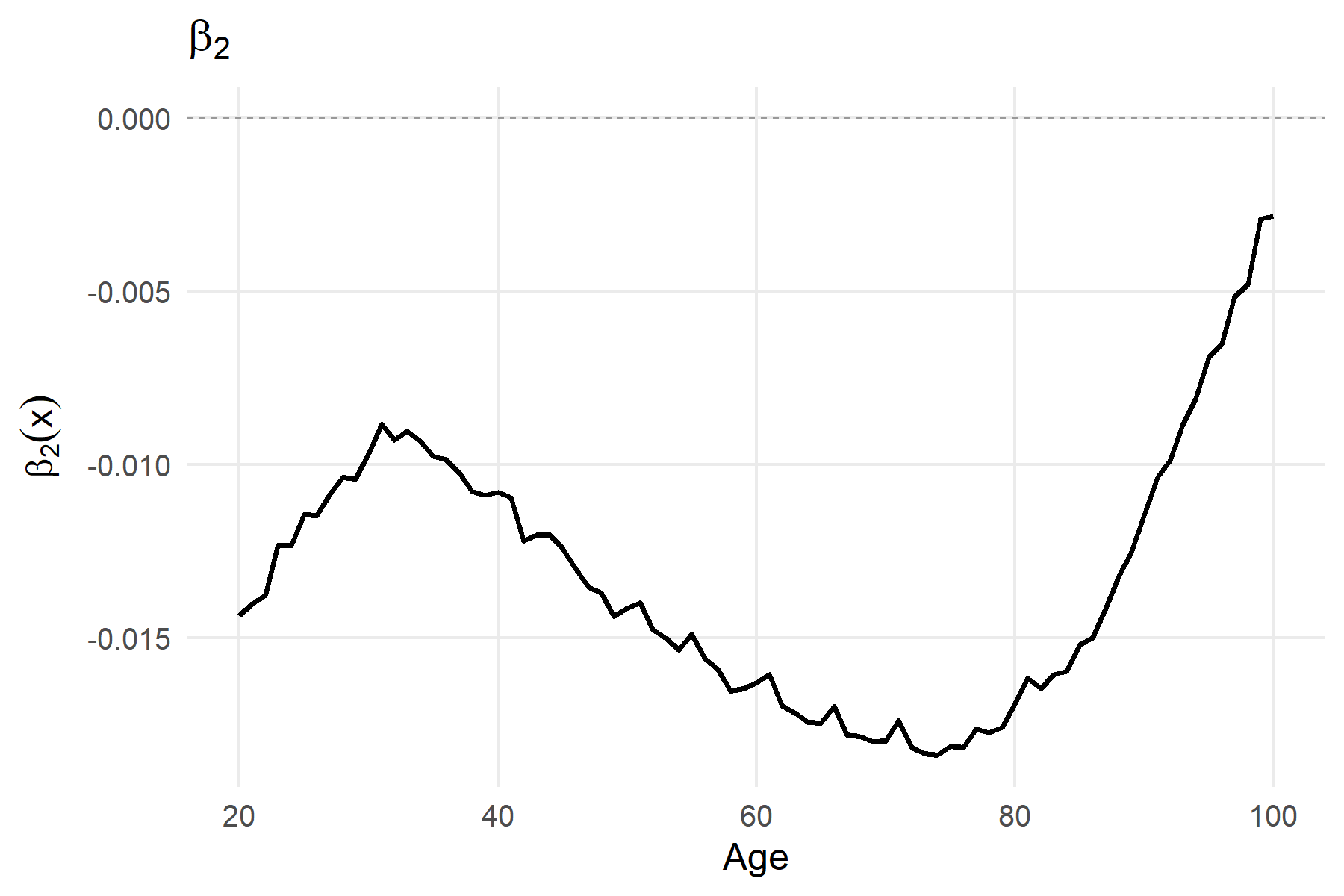}
    \caption{Age loading $\beta_{2}(x)$}
    \label{fig:beta2}
\end{subfigure}
\begin{subfigure}{0.48\textwidth}
    \centering
    \includegraphics[width=\linewidth]{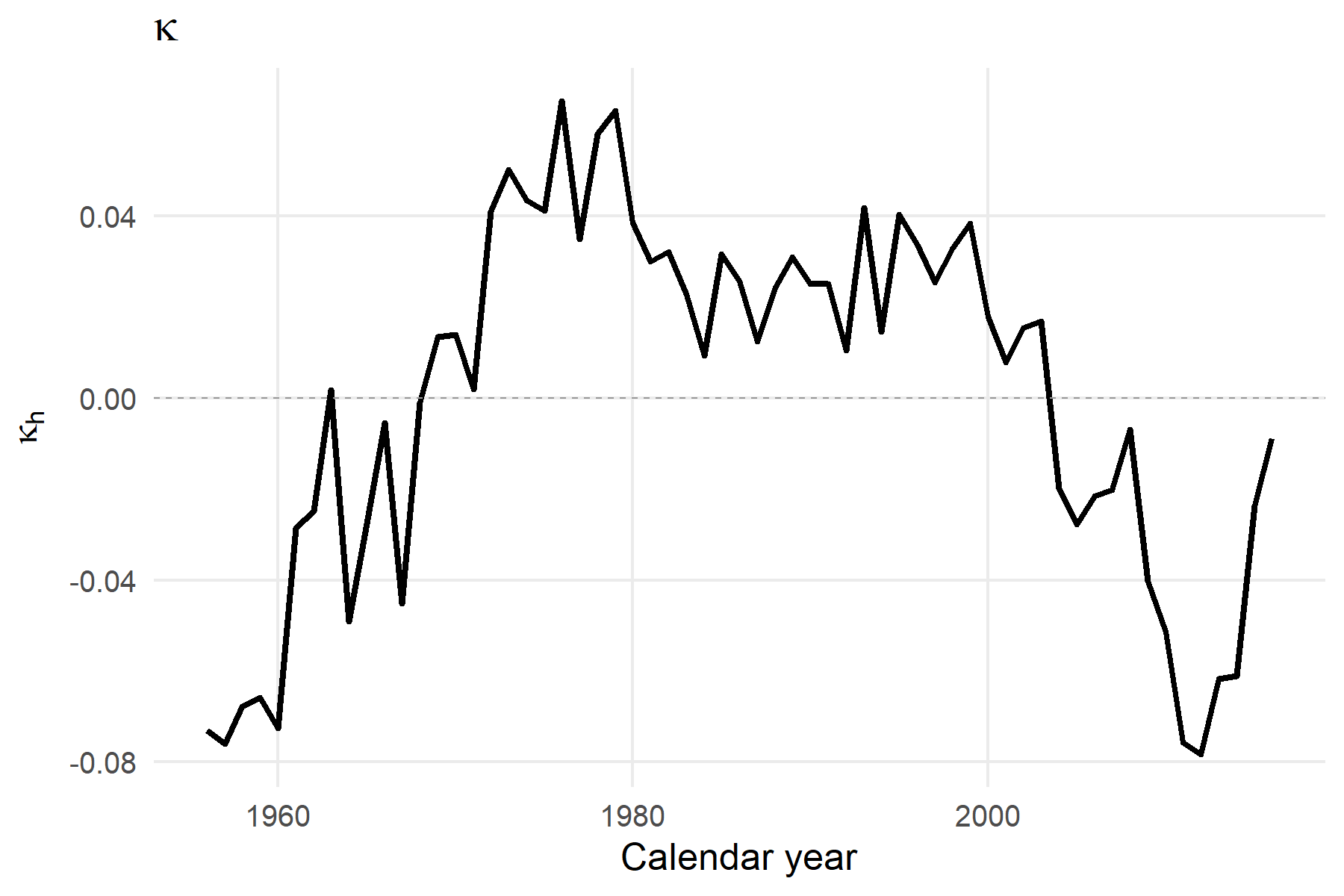}
    \caption{Period effect $\kappa_h$}
    \label{fig:kappa_single}
\end{subfigure}\hfill
\begin{subfigure}{0.48\textwidth}
    \centering
    \includegraphics[width=\linewidth]{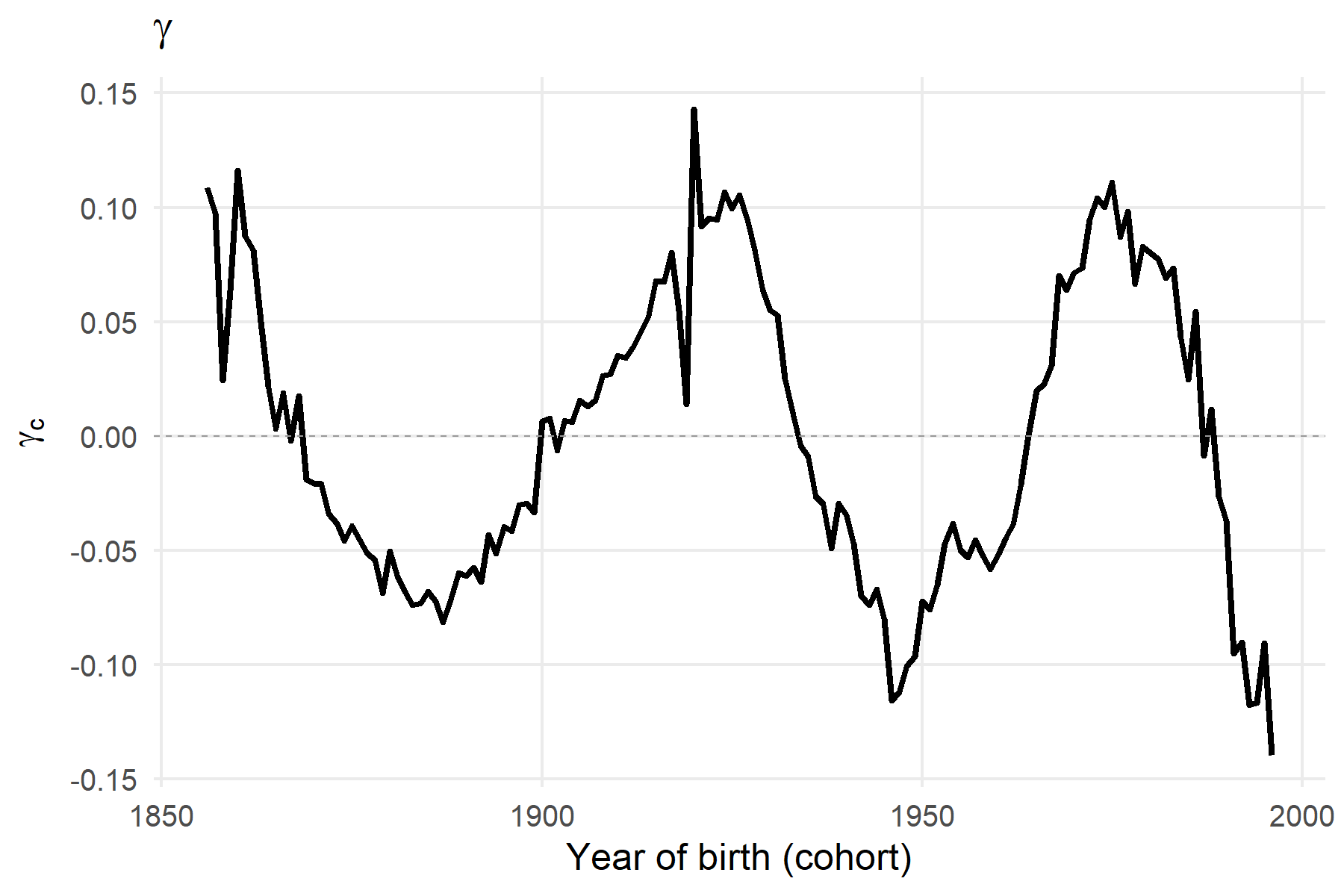}
    \caption{Cohort effect $\gamma_c$}
    \label{fig:gamma_single}
\end{subfigure}
\caption{Estimated components of the single-population APCI model.
Subfigures (a) and (b) display the estimated age effects $\beta_1(x)$ and $\beta_2(x)$,
while (c) and (d) present the period and cohort factors $\kappa_h$ and $\gamma_c$.}
\label{fig:apci_components_single}
\end{figure}

\subsection{Two-population gravity model}

Figure~\ref{fig:apc_components} shows the estimated factors in the two-population gravity model, and Table~\ref{tab:grav_params_compact} summarizes the parameters of the $\kappa$ dynamics in Equation~\eqref{eq:kappa_2}.

\begin{figure}[H]
\centering
\captionsetup[subfigure]{labelformat=parens,labelsep=space,font=small}
\begin{subfigure}{0.48\textwidth}
    \centering
    \includegraphics[width=\linewidth]{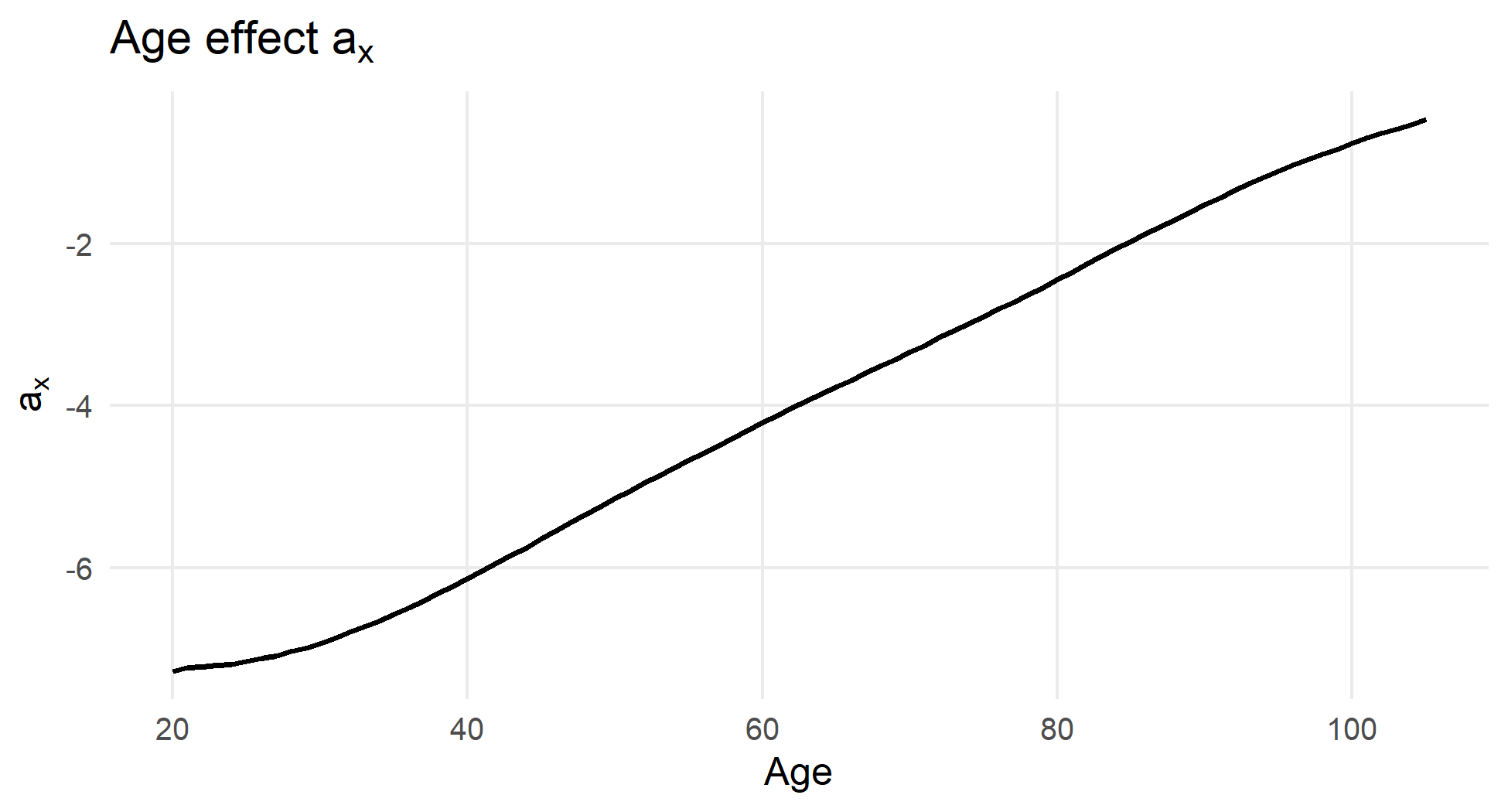}
    \caption{Age effect $\alpha_x$}
    \label{fig:alpha_x}
\end{subfigure}\hfill
\begin{subfigure}{0.48\textwidth}
    \centering
    \includegraphics[width=\linewidth]{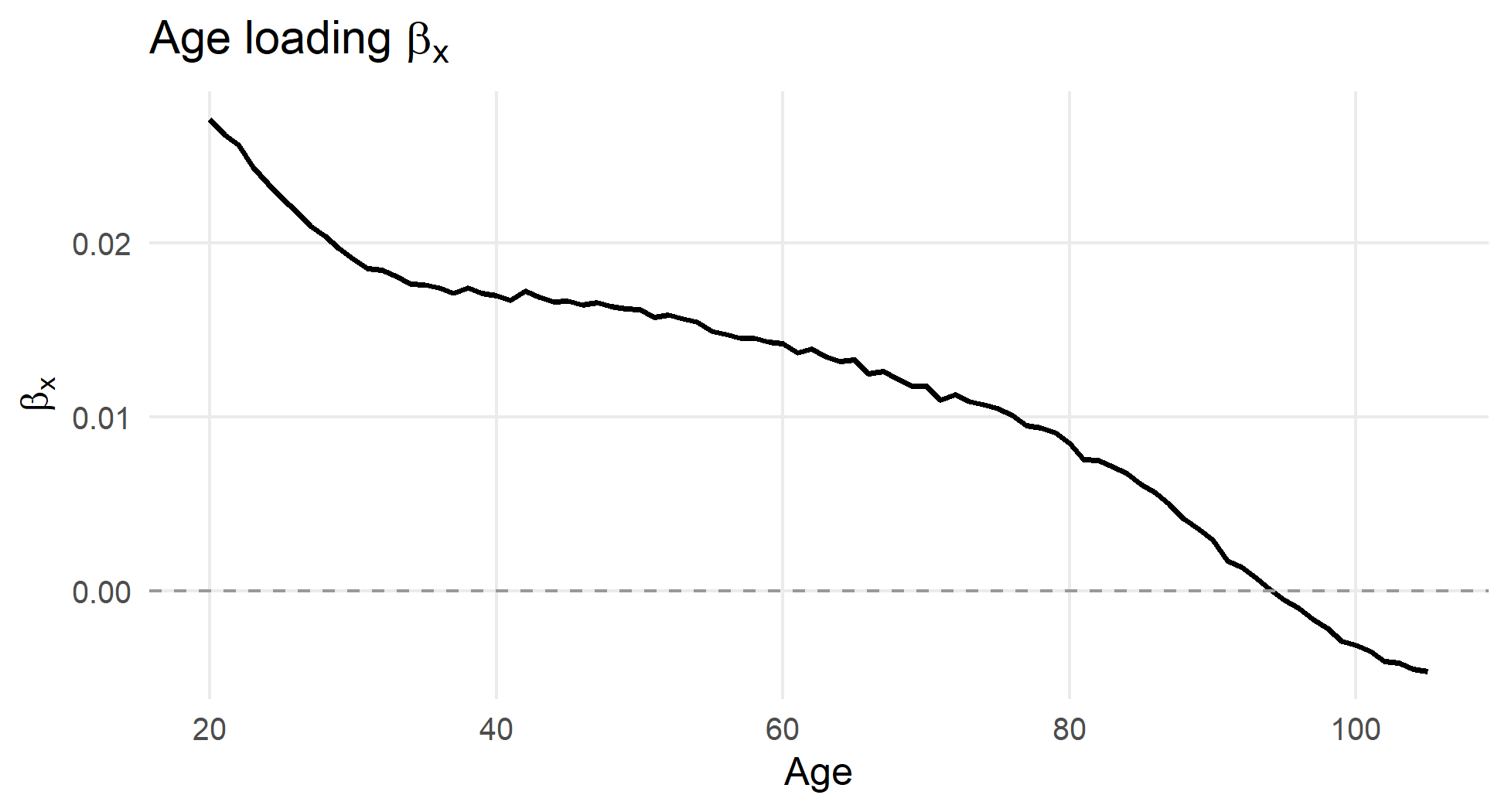}
    \caption{Age loading $\beta_x$}
    \label{fig:beta_x}
\end{subfigure}
\begin{subfigure}{0.48\textwidth}
    \centering
    \includegraphics[width=\linewidth]{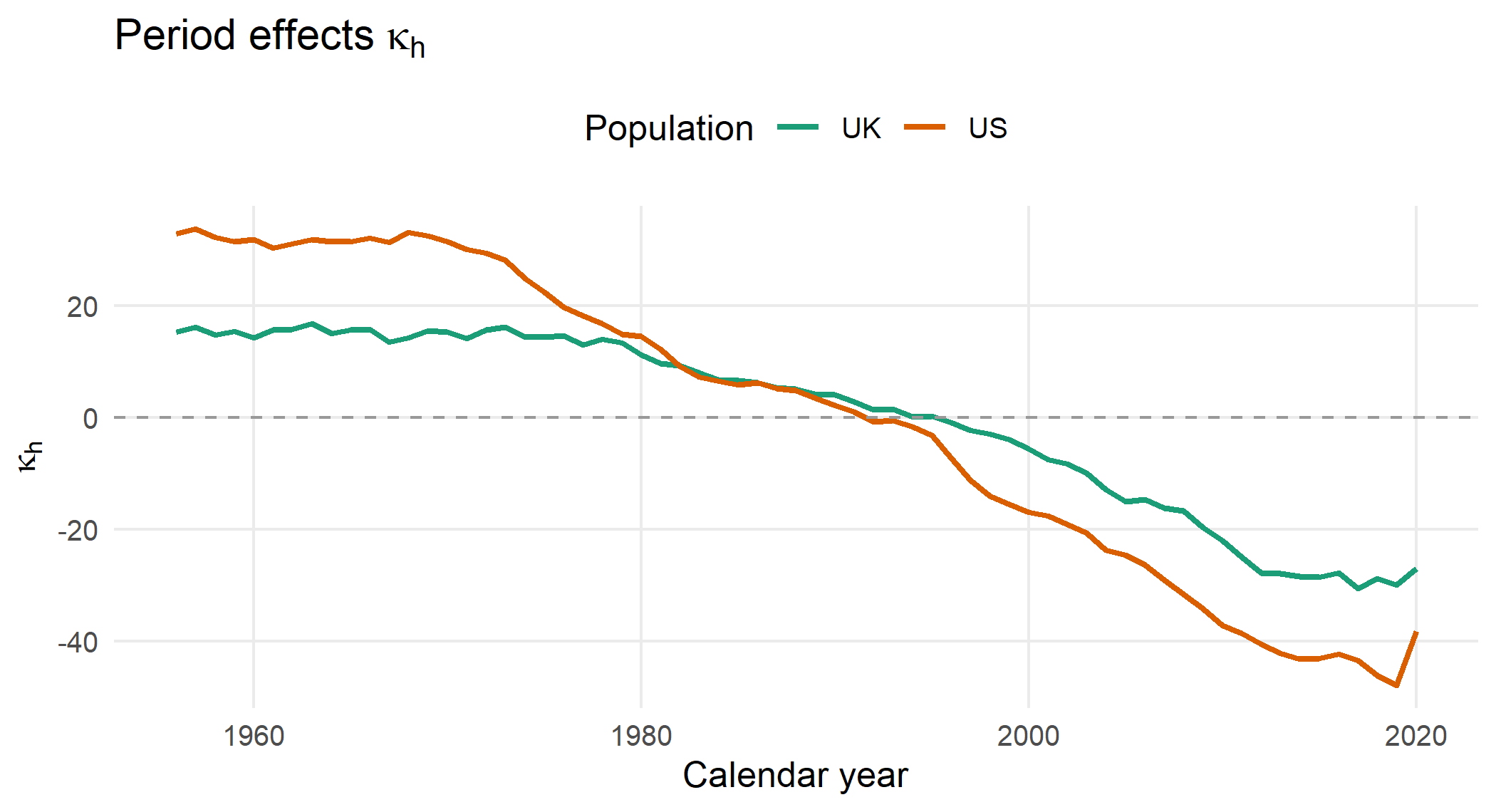}
    \caption{Period effects $\kappa_h$}
    \label{fig:kappa}
\end{subfigure}\hfill
\begin{subfigure}{0.48\textwidth}
    \centering
    \includegraphics[width=\linewidth]{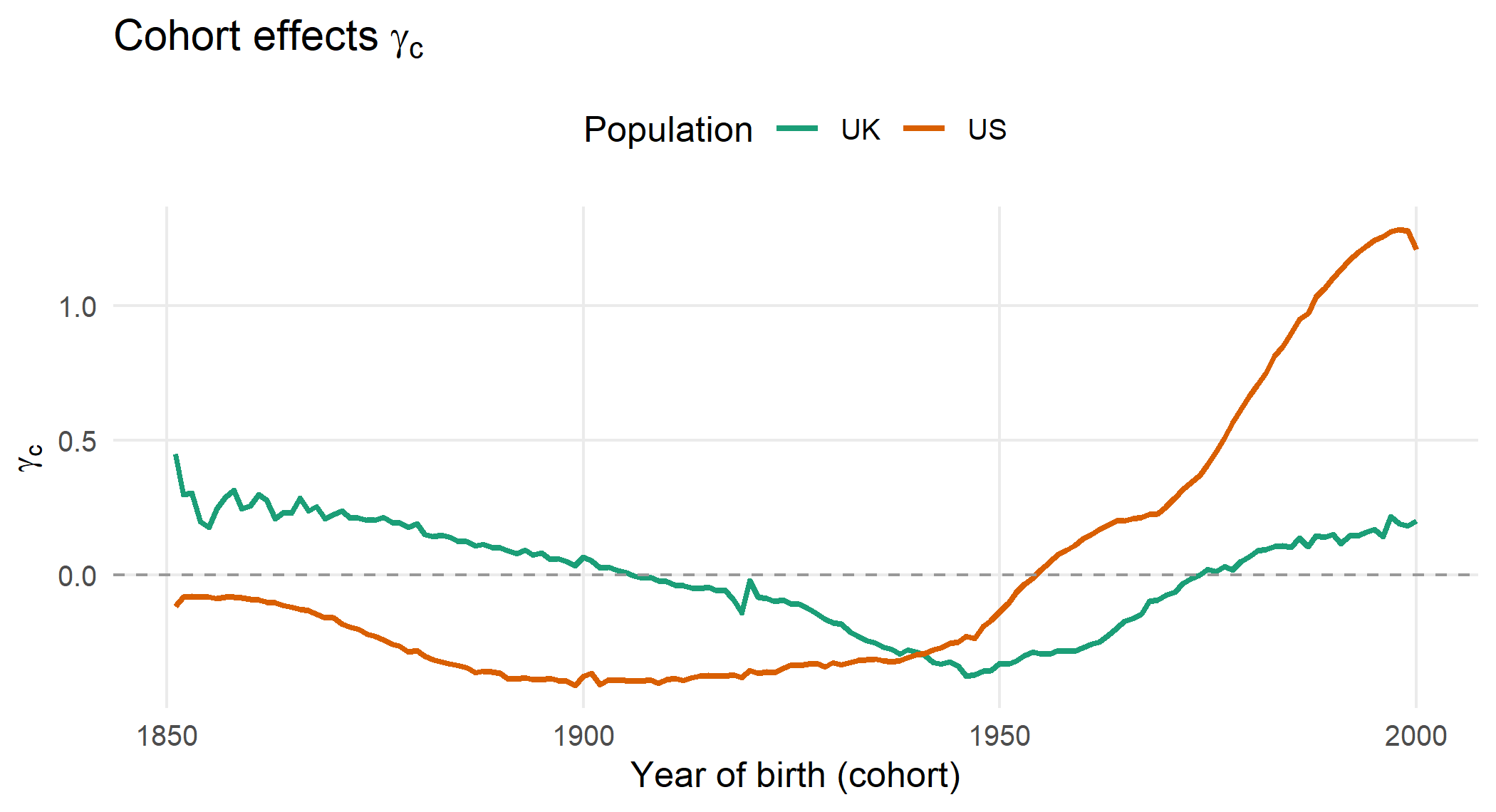}
    \caption{Cohort effects $\gamma_c$}
    \label{fig:gamma}
\end{subfigure}
\caption{Estimated components of the two-population APC gravity model.
Subfigures (a) and (b) show the common age effects ($\alpha_x$, $\beta_x$),
while (c) and (d) depict the period and cohort factors ($\kappa_h$, $\gamma_c$)
for the UK and US populations.}
\label{fig:apc_components}
\end{figure}

\begin{table}[H]
\centering
\caption{Estimated period ($\boldsymbol{\kappa}$) dynamics parameters in the two-population gravity model}
\label{tab:grav_params_compact}
\begin{tabular}{l r @{\hspace{2em}} l r}
\toprule
\multicolumn{4}{c}{\textbf{Period ($\boldsymbol{\kappa}$) dynamics}} \\
\midrule
$\mu_{\kappa}^{(1)}$   & $-0.6606$ & $\mu_{\kappa}^{(2)}$   & $-1.1085$ \\
$\sigma_{\kappa}^{(1)}$& $1.2761$  & $\phi_{\kappa}$        & $0.0103$  \\
$\sigma_{\kappa,21}$   & $0.9532$  & $\sigma_{\kappa,22}$   & $1.5709$  \\
\bottomrule
\end{tabular}
\end{table}

\end{document}